\pdfoutput=1
\newif\ifFull
\Fulltrue

\newif\ifMapReduce
\MapReducefalse

\newif\ifAllPairs
\AllPairsfalse

\newif \ifStrata
\Stratafalse

\ifFull
\documentclass[11pt]{article}
\usepackage[margin=1in]{geometry}
\usepackage{fullpage}
\else
\documentclass{cccg13}
\fi
\usepackage{graphicx,amssymb,amsmath}
\usepackage{url}
\usepackage{enumerate}
\usepackage{epstopdf}
\ifFull
\usepackage{algorithmic}
\else
\usepackage[noend]{algorithmic}
\fi
\usepackage{hyperref}

\ifFull
\newtheorem{theorem}{Theorem}%[section] uncomment to number theorems by section
\newtheorem{lemma}[theorem]{Lemma}

\newtheorem{corollary}[theorem]{Corollary}

\newtheorem{remark}[theorem]{Remark}

\newenvironment{proof}{\noindent{\bf Proof:}}{\hspace*{\fill}\rule{6pt}{6pt}\bigskip}
\else

\newtheorem{remark}[theorem]{Remark}
\fi

\makeatletter
\def\@begintheorem#1#2{\sl \trivlist \item[\hskip \labelsep{\bf #1\ #2:}]}
\def\@opargbegintheorem#1#2#3{\sl \trivlist
      \item[\hskip \labelsep{\bf #1\ #2\ #3:}]}
\makeatother

\title{
Set-Difference Range Queries
}
\author{
David Eppstein
\thanks{
Dept. of Comp. Sci.,
U.~of CA, Irvine, 
\texttt{eppstein(at)uci.edu}
}
\and
Michael T. Goodrich 
\thanks{
Dept. of Comp. Sci.,
U.~of CA, Irvine, 
\texttt{goodrich(at)uci.edu}
}
\and
Joseph A. Simons
\thanks{
Dept. of Comp. Sci.,
U.~of CA, Irvine, 
\texttt{jsimons(at)uci.edu}
}
}
\date{}

\begin{document}
\thispagestyle{empty}
\maketitle 

\begin{abstract}
We introduce the problem of performing 
set-difference range queries, where answers to queries are
set-theoretic symmetric differences between 
sets of items in two geometric ranges. 
We describe a general framework for answering
such queries based on a novel use of 
data-streaming sketches we call 
\emph{signed symmetric-difference} sketches.
We show that such sketches can be realized using
invertible Bloom filters (IBFs), which
can be composed, differenced, 
and searched so as to solve set-difference range queries in a wide
range of scenarios.
% including orthogonal range queries, simplex range
% queries, sub-image range queries, and interval stabbing queries.
% Our approach allows for query times whose time performance depends on
% the size of the set difference in the range, rather than on the size of the union of the two queried ranges.
%\textbf{Set-difference range queries}:
%preprocess a collection of sets of points in $\R^d$, for fixed $d\ge1$,
%so that given two geometric ranges, $R_1$ and $R_2$,
%and two sets, $S_1$ and $S_2$, one can quickly
%report (or count) the points in the
%set-theoretic symmetric difference 
%$(R_1\cap S_1)\ \Delta\ (R_2\cap S_2)$.
%\ifAllPairs
%\end{itemize}
%We show how to solve each of the above problems efficiently
%using our randomized data structures, which include invertible Bloom filters for
%listing the elements of a set difference and strata estimators for approximating
%the difference size. Our constructions involve minor modifications to the
%previous data structures in order to reduce their failure rate to an arbitrarily
%small parameter, as is needed in these algorithmic applications.
%\fi{}
\end{abstract}

\ifFull{
\setcounter{page}{0}
\clearpage
}\fi{}
%
%
% paper body
\section{Introduction}

Efficiently
identifying or quantifying the differences between two sets is a problem that arises frequently in applications, for example, when clients synchronize the
calendars
on their smart phones with their calendars at work, when
databases reconcile their contents after a network partition, or
when a backup service queries a file system for any changes that have
occurred since the last backup.
Such queries can be global, for instance,
in a request for the differences across all data values for a pair of
databases, or they can be localized, requesting differences for
a specific range of values of particular interest.
For example,
clients might only need to synchronize their calendars for a certain
range of dates or a pair of databases may
need only to reconcile their contents for a certain range of
transactions.
We formalize this task by a novel type of range searching problem, which we call \emph{set-difference range queries}.

We assume a collection $\cal X$ of sets $\{X_1,X_2,\ldots,X_N\}$, containing \emph{data items} that are each associated with a geometric point and with a member of
universe, $\cal U$, of size $U=|{\cal U}|$. A set-difference range query is specified by the indices of a pair of data sets,
$X_i$ and $X_j$, and by a pair of \emph{ranges}, $R_1$ and $R_2$, which are each a
constant-size description of a set of points such as a
hyper-rectangle, simplex, or half-space.
The answer to this set-difference range query consists of
the elements of $X_i$ and $X_j$ whose associated points belong to the ranges $R_1$ and $R_2$ respectively, and whose associated elements in $U$ are
contained in one of the two data sets but not both. 
Thus, we preprocess 
$\cal X$
%a collection of sets of points in $\R^d$, for fixed $d\ge1$,
so that given ranges, $R_1$ and $R_2$,
and two sets, $X_1$ and $X_2$, we can quickly
report (or count) the universe elements in the
set-theoretic symmetric difference 
$(R_1\cap X_1)\ \bigtriangleup \ (R_2\cap X_2)$.
%That is, the answer is the set $(X_i \cap R_1)
%\bigtriangleup (X_j \cap R_2)$,
%where $\bigtriangleup$ denotes symmetric difference.
The performance goal in answering such
queries is to design data structures having low space and
preprocessing requirements that
support fast set-difference range queries whose time depends
primarily on the size of the difference, not the number of items in
the range 
%Including figure in the body.
%\ifFull
    (see Figure~\ref{fig:differences}).
%\else
%    (see Figure~\ref{fig:differences} in Appendix~\ref{app:figs}).
%\fi{}
%
%\ifFull
Examples of such scenarios include the following.
\begin{figure}[t!]
\centering
\includegraphics[width=3.5in, trim = 0.1in 0.75in 0.25in 0.75in, clip]{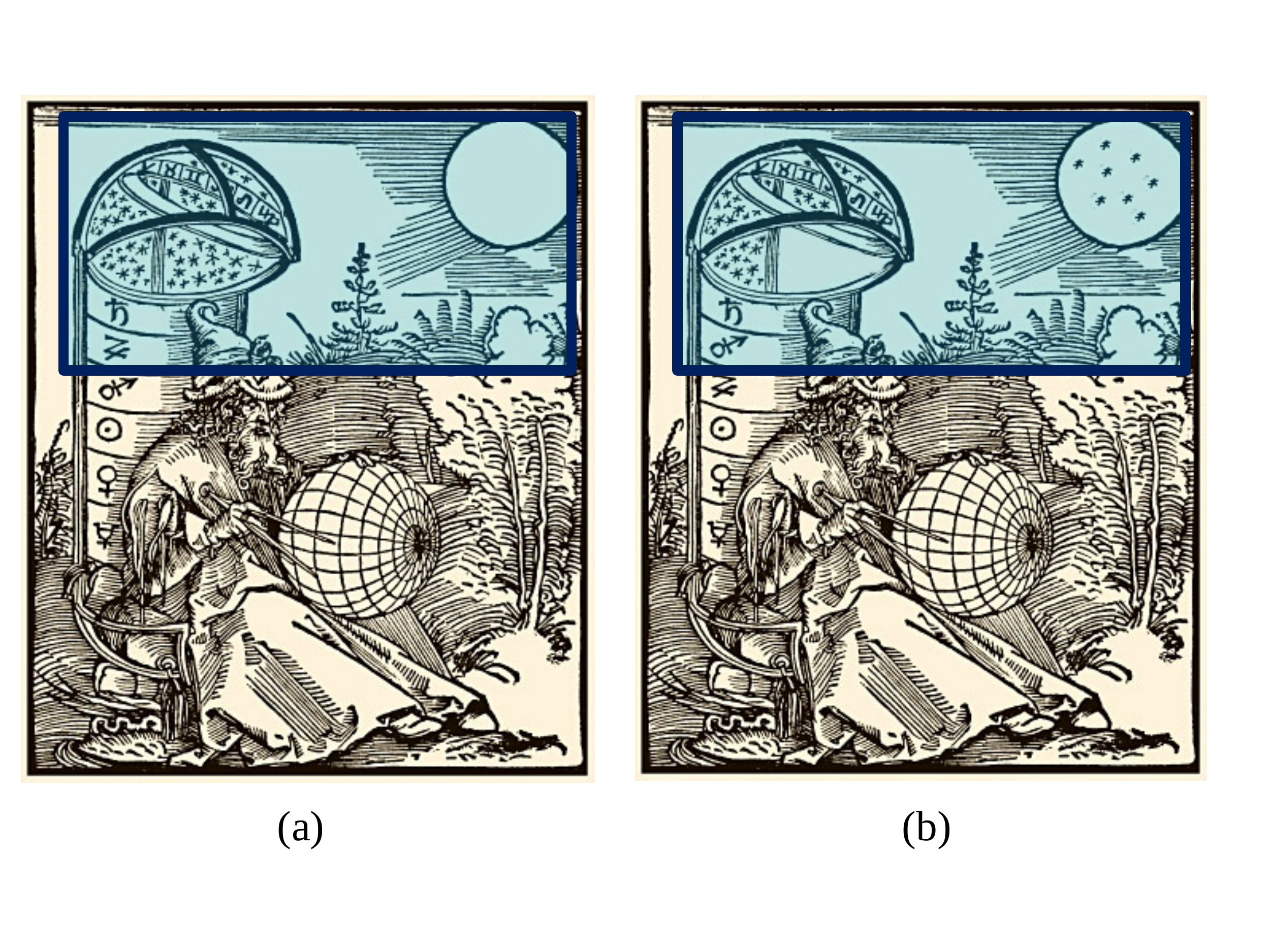}
\vspace{-20pt}
\caption{\label{fig:differences} Illustrating the set-difference
range query problem. The images in (a) and (b) have four major
differences, three of which are inside the common query range.
The image (a) is a public-domain engraving of an astronomer by Albrecht
D{\"u}rer, from the title page of 
\textit{Messahalah, De scientia motus orbis} (1504).}
\end{figure}
%\fi{}
%

\begin{table*}[hbt]
\vspace*{-12pt}
\centering
\begin{tabular}{|l l|l|l|} 
\hline
%		\multicolumn{5}{|c|}{Range Query Summary Table}   \\ \hline
		Query Type   &                   &  Query Time         & Space   \\ \hline
Orthogonal:	&	Standard~\ifFull\cite{Lueker1978,Bentley1980}\else\cite{Bentley1980}\fi{}  & $O(\log^{d-1} n)$ & $O(n \log^{d-1} n)$     \\	
     %   &  SD naive                      & $O(n + \log^{d-1} n)$      & $O(n \log^{d-1} n)$     \\
		&  SD fixed $m$                  & $O(m \cdot \log ^d n)$     &  $O(m \cdot n \log^{d-1} n)$     \\
		&  SD variable $m$               & $O(m \cdot \log ^d n)$     &  $O(n \log ^d n)$ 		\\ 
		&  SD size est.                  & $O(\log ^{d+1} n \log U)$         & $O(n \log^{d} n \log U)$         \\ \hline
Simplex:& Standard~\cite{Matousek1992}   & $O(n^{1-1/d}(\log n)^{O(1)})$   & $O(n)$                 \\
	%	& SD naive                       &  $O(n+ n^{1-1/d}(\log n)^{O(1)})$    & $O(n)$                 \\
        & SD fixed $m$                   & $O(m \cdot n^{1-1/d}(\log n)^{O(1)})$   & $O(m\cdot n)$             \\
		& SD variable $m$                & $O(m \cdot n^{1-1/d}(\log n)^{O(1)})$   & $O(n \log \log n)$        \\ 		
		& SD size est.                   & $O(n^{1-1/d}(\log n)^{O(1)} \log U)$ & $O(n \log n \log U )$          \\ \hline
Stabbing:& Standard~\cite{Bentley1977}   &  $O(\log n)$         &     $O(n \log n)$            \\
	%	& SD naive                       &  $O(n + \log n)$         &     $O(n \log n)$            \\
        & SD fixed $m$                   &  $O(m \cdot \log n)$ &     $O(m \cdot n \log n)$    \\
		& SD variable $m$                &  $O(m \cdot \log n)$ &     $O( n \log n)$           \\ 		
		& SD size est.                   &  $O(\log^2 n \log U)$       &     $O(n \log n \log U)$          \\ \hline     
%\ifFull
Partial Sum:& Standard\footnotemark &       $O(1)$              & $O(n)$           \\
	%	& SD naive                       &       $O(n)$              & $O(n)$          \\
        & SD fixed $m$                   &       $O(m)$              & $O(m\cdot n)$    \\
		& SD variable $m$                &       $O(m)$              & $O(n^2)$         \\ 			
        & SD size est.                   &       $O(\log n \log U)$         & $O(n \log n \log U)$    \\ \hline
%\fi
\end{tabular}

\medskip
\caption{
	%Results when applying 
	%Theorems~\ref{thm:fixed-cardinality-reporting},
	%\ref{thm:variable-cardinality-reporting}, 
	%and~\ref{thm:size-estimation} to several range query data structures.
	The results labeled ``Standard'' are previously known results 
	for each data structure. 
	Results labeled 
	``SD'' indicate bounds for set-difference range queries. 
    %In the results labeled ``SD'' (set-difference), 
    %%% Naive is O(n) no matter what
    %we compare the naive approach with our 
%results. %based on signed-symmetric difference structures.
Here $d$ is the dimension of the query, $m$ is the output size, and we assume
the approximation factor
$(1 \pm \epsilon)$ and failure probability $\delta$ are fixed.  
}
	\ifFull \label{tab:summaryFull} \else \label{tab:summary} \fi 
\end{table*}

\begin{itemize}
\item
Each set contains readings from a group of sensors in a given
time quantum
(e.g., see~\cite{Basseville1988309}).
Researchers may be interested in determining which
sensor values have changed between two time quanta in a given
region.
\item
Each set is a catalog of astronomical objects
in a digital sky survey. 
Astronomers are often interested in objects that appear or disappear
in a given rectangular region between pairs 
of nightly observations.
(E.g., see~\cite{Sloan}.)
\item
Each set is an image taken at a certain time and place.
Various applications may be interested in pinpointing changes that
occur between pairs of images 
(e.g., see~\cite{raar-icda-05}), 
which is something that might be done, for instance,
via two-dimensional binary search and repeated set-difference range
queries.
\end{itemize}

%hack to make footnote in table work
\makeatletter
\global\let\Hy@saved@currentHref\@currentHref
\hyper@makecurrent{Hfootnote}%
\global\let\Hy@footnote@currentHref\@currentHref
\global\let\@currentHref\Hy@saved@currentHref
\makeatother
\footnotetext{A standard solution to the partial sum problem is illustrated in
Figure \ref{figPrefixSums}}

\subsection{Related Work}
We are not aware of prior work on set-difference range
queries. 
However,
Suri \emph{et al.}~\cite{Suri2006} consider approximate range counting in 
data streams.
Shi and JaJa~\cite{sj-anfat-05} present a data structure for 
range queries indexed at a specific time
for data that is changing over time, achieving
polylogarithmic query times.
If used for set-difference range queries, however, their scheme would
not produce answers in time proportional
to the output size.
\ifFull
Such queries are also related to 
kinetic range-query data structures, 
such as those given by Agarwal \textit{et al.}~\cite{aae-imp-03}
or \v{S}altenis \textit{et al.}~\cite{sjll-ipcm-00},
which can answer range queries
for a given moment in time 
for points moving along known trajectories.
For a survey of various schemes for indexing and search temporal
data, please see Mokbel \textit{et al.}~\cite{Mokbel03spatio-temporalaccess}.
\fi%
For a survey of general schemes for range searching data
structures, see Agarwal~\cite{a-rs-97}.

\subsection{Our Results}
We provide general methods for supporting a wide class of set-difference range queries
by combining signed-symmetric difference sketches 
with any canonical group or
semigroup range searching structure.
Our methods solve range difference queries where the two sets being
compared may be drawn from the same or different data sets, and may be defined
by the same or different ranges. In our data structures, sets are combined in
the \emph{multiset model}: two data items may be associated with the same
element of $U$, and if they belong to the same query range they are considered
to have multiplicity two, while if they belong to the two different query ranges
defining the set difference problem then their cardinalities cancel. 
The result
of a set difference query is the set of all elements whose total cardinality
defined in this way is nonzero. 

Our data structures are probabilistic and return
the correct results with high probability. Our running times depend on the size
of the output, but only weakly depend on the size of the original sets. 
In particular, we derive the results
shown in Table~\ifFull \ref{tab:summaryFull} \else \ref{tab:summary}\fi{}
for the following range-query problems 
\ifFull
(see Section~\ref{sec:sdDetails} for details):
\else
(see Appendix~\ref{sec:sdDetails} for details):
\fi

\begin{itemize}
\itemsep0pt
\item Orthogonal:
Preprocess a set of points in $R^d$
 such that given a query range defined by an axis-parallel hyper rectangle, we
 can efficiently report or estimate the size of the set of points contained in the hyper-rectangle
\item
Simplex:
Preprocess a set of points in $R^d$ such that given a query range defined 
by a simplex in $R^d$, we can efficiently report or estimate the size of set of points contained in the simplex
\item
Stabbing:
Preprocess a set of intervals such that given a query point $x$, we can efficiently
report or estimate the size of the subset which intersects $x$.
\item
Partial Sum:
Preprocess  a grid of total size $O(n)$ (e.g. an image with $n$ pixels, or a $d$
dimensional array with $O(n)$ entries for any constant $d$) such that given a
query range defined by a hyper-rectangle on the grid, we can efficiently report
or estimate the sum of the values contained in that hyper-rectangle.
\end{itemize}

\ifFull
In the remainder of the paper, we show
how to perform set-difference range reporting and counting, based on
the framework of using signed-symmetric difference reporting (SDR)
and signed-symmetric difference counting (SDC) data
structures.
Finally, we discuss the underlying methods for implementing
SDR and SDC sketches in more detail in Section~\ref{app:IBF} and Section~\ref{sec:fme}, respectively, based on the use of invertible Bloom filters and frequency
moments.
\fi

%old location of prefix sums figure

%\subsection{The Abstract Range Searching Framework} \label{secFramework}
\section{The Abstract Range Searching Framework} \ifFull \label{secFrameworkFull} \else \label{secFramework} \fi{} 
In order to state our results in their full generality it is necessary to
provide some general definitions from the theory of range searching (e.g.,
see~\cite{a-rs-97}).

A \emph{range space} is a pair $(X,R)$ where $X$ is a universe of objects that
may appear as data in a range searching problem (such as the set of all points
in the Euclidean plane) and $R$ is a family of subsets of $X$ that may be used
as queries (such as the family of all disks in the Euclidean plane).  A
\emph{range searching problem} is defined by a range space together with an
\emph{aggregation function} $f$ that maps subsets of $X$ to the values that
should be returned by a query. For instance, in a \emph{range reporting
problem}, the aggregation function is the identity; in a \emph{range counting
problem}, the aggregation function maps a subset of $X$ to its cardinality. A
data structure for the range searching problem must maintain a finite set
$Y\subset X$ of objects, and must answer queries that take a range $r\in R$ as
an argument and return the value $f(r\cap Y)$. The goal of much research in
range searching is to design data structures that are efficient in terms of
their space usage, preprocessing time, update time (if changes to $Y$ such as
insertions and deletions are allowed), and query time. In particular, it is
generally straightforward to answer a query in time $O(|Y|)$, by testing each
member of $Y$ for membership in the query range $r$; the goal is to answer
queries in an amount of time that may depend on the output size but that does
not depend so strongly on the size of the data set.

A range searching problem is said to be \emph{decomposable} if there exists a
commutative and associative binary operation $\oplus$ such that, for any two
disjoint subsets $A$ and $B$ of $X$, $f(A\cup B)=f(A)\oplus f(B)$. In this case,
the operation $\oplus$ defines a \emph{semigroup}.

Many range searching data structures have the following form, which 
we call a
\emph{canonical semigroup range searching data structure}:  %there's already a name for this
the data structure
stores the values of the aggregation function on some family $F$ of sets of data
values, called \emph{canonical sets}, with the property that every data set
$r\cap Y$ that could arise in a query can be represented as the disjoint union
of a small number of canonical sets $r_1$, $r_2$, $\dots$, $r_K$. To answer a
query, a data structure of this type performs this  decomposition of the query
range into canonical sets, looks up the values of the aggregation function on
these sets, and combines them by the $\oplus$ operation to yield the query
result. Note that this is sometimes called a \emph{decomposition scheme}.
 For instance, the order-statistic tree is a binary search tree over an
ordered universe in which each node stores the number of its descendants
(including itself) in the tree, and it can be used to quickly answer queries
that ask for the number of data values within any interval of the ordered
universe. This data structure can be seen as a canonical semigroup range
searching data structure in which the aggregation function is the cardinality,
the combination operation $\oplus$ is integer addition, and the canonical sets
are the sets of elements descending from nodes of the binary search tree. Every
intersection of the data with a query interval can be decomposed into $O(\log
n)$ canonical sets, and so interval range counting queries can be answered with
this data structure in logarithmic time.

\ifFull
In general, there is a tradeoff between $|F|$,
the number of canonical subsets stored by the data
structure, and the number of canonical subsets required to answer a
query.  To illustrate, we consider a range searching problem in one dimension. 
If we have $n$ canonical subsets, then in order to cover all
possible queries we need to have each subset be a singleton set, and to answer a
query we may need to use $O(n)$ subsets. If we use $O(n \log n)$ canonical subsets,
then we can answer a query using only $O(\log n)$ of them as described above.
If we use $O(n^2)$
canonical subsets, we only need to use $O(1)$ canonical subsets to answer a query,
because there are only $O(n^2)$ unique answers to a range query in one
dimension, and by using $O(n^2)$ canonical subsets we can precompute all of
them.
\fi{}
%TODO: ok to remove these paragraphs to make space?
When the combination operation $\oplus$ has additional properties, they may
sometimes be used to obtain greater efficiency. In particular, if $\oplus$ has
the structure of a group, then we may form a \emph{canonical group range
searching data structure}. Again, such a data structure stores the values of the
aggregation function on a family of sets of data values, but in this case it
represents the query value as an expression $(\pm f(r_1))\oplus (\pm
f(r_2))\oplus\cdots$, where the canonical sets $r_i$ and their signs are chosen
with the property that each element of the query range $r$ belongs to exactly
one more positive set than negative set. Again, in the interval range searching
problem, one may store with each element its \emph{rank}, the number of elements
in the range from $-\infty$ to that element, and answer a range counting query
by subtracting the rank of the right endpoint from the rank of the left
endpoint. In this example, the ranks are not easy to maintain if elements are
inserted and deleted, but they allow interval range queries to be answered by
combining only two canonical sets instead of logarithmically many.

\subsection{Signed Symmetric-Difference Sketches}
\label{sec:sd-sketches}
Suppose that we want to represent an input set $S$ in space sub-linear in the size of $S$
such that we can compute some function $f(S)$ on our compressed representation.
This problem often comes up in the streaming literature, and common solutions include 
dimension reduction (e.g. by a Johnson-Lindenstrauss transform~\cite{kn-sjlt-12}) and
computing a \emph{sketch} of $S$ (e.g. Count-Min Sketch~\cite{CM-Sketch}). 

A \emph{sketch} $\sigma_S$ of a set $S$ is a randomized compressed representation of $S$ 
such that we can approximately and probabilistically compute $f(S)$ by evaluating an
appropriate function $f'(\sigma_S)$ on the sketch $\sigma_S$.
%Ideally, the size of the sketch should only be polylogarithmic in the size of the original set.
This construct also comes up when
handling massive data data sets, and in this context the compressed representation is
sometimes called a synopsis~\cite{gm-sdsmds-99}.

A sketch algorithm $\sigma$ is called \emph{linear} if it 
has a group structure. That is, there exist two operators $\oplus$ and $\ominus$ on
sketches $\sigma$ such that given two multi-sets $S$ and $T$,
\[
\sigma_{S \uplus T} = \sigma_S \oplus \sigma_T \mbox{ and } \sigma_{S \setminus T} = \sigma_S \ominus
\sigma_T
\]
where $\uplus$ and $\setminus$ are the multi-set addition and subtraction operators
respectively.

For our results, we define two different types of linear sketches. A \emph{Signed
Symmetric-Difference Reporting} (SDR) sketch is a linear sketch that supports a 
function \texttt{report}: given a pair of sketches $\sigma_S$ and $\sigma_T$ 
for two sets $S$ and $T$ respectively, probabilistically compute 
$S \setminus T$ and $T \setminus S$ using only information stored in the sketches
$\sigma_S$ and $\sigma_T$ in 
$O(1 + m)$ time, where $m$ is the cardinality of the output.
A \emph{Signed Symmetric-Difference Cardinality} (SDC) sketch is a linear sketch that
supports a function \texttt{count}: given a pair of sketches $\sigma_S$ and $\sigma_T$ for
two sets $S$ and $T$ respectively, probabilistically approximate $|S \bigtriangleup T|$
using only information stored in the sketches $\sigma_S$ and $\sigma_T$ in time linear in
the size of the sketches.

\section{Main  Results} %TODO better name 
The main idea of our results is to represent each canonical set by 
an SDR or an SDC sketch. 
We implement our signed symmetric-difference counting sketches using a linear sketch based
on the frequency moment estimation techniques of Thorup and Zhang~\cite{tz-tb4uh-04}.
We implement our signed symmetric-difference reporting sketches via an
\emph{invertible Bloom filter} (IBF), a data structure introduced for straggler
detection in data streams~\cite{eg-sirtd-11}. 
%and also used for
%TODO citations
%distributed version control~\cite{eguv-esswp-11}.
% I don't see the point of citing this one -- as far as I can tell it's dead. DE
%, and efficient computation of
%all-pairs Hamming distances~\cite{Eppstein:2011}. 
%
IBFs can be added and
subtracted, giving them a group structure and allowing an IBF for a query range
to be constructed from the IBFs for its constituent canonical sets. The
difference of the IBFs for two query ranges is itself an IBF that allows the difference
elements to be reported when the difference is small. To handle set differences of varying
sizes we use a hierarchy of IBFs of exponentially growing sizes, together with some
special handling for the case that the final set difference size is larger than the size
of some individual canonical set. 

Further details of the SDR and SDC sketches are given in
later sections.
In this section, we assume the existence of SDR and SDC
sketches as defined above in order to prove the following three theorems which are the
crux of our results listed in 
\ifFull
Table~\ref{tab:summaryFull}.
\else
Table~\ref{tab:summary}
%omit these for now
\fi

\begin{theorem}
\ifFull \label{thm:fixed-cardinality-reportingFull} \else \label{thm:fixed-cardinality-reporting} \fi{} 
Suppose that a fixed limit $m$ on the cardinality of the returned set
differences is known in advance of constructing the data structure, and our
queries must either report the difference if it has cardinality at most $m$, or otherwise report that it is too large.
In this case, we can answer set-difference range queries with probability at least $1 - \epsilon$
for any range space that can be modeled by a 
canonical group or semigroup range searching data structure.
Our solution stores $O(m)$ words of aggregate information per canonical set,
uses a combination operation $\oplus$ that takes time $O(m)$ per combination,
and allows the result of this combination to be decoded to yield the query
results in $O(m)$ time. If the data structure is updated, the aggregate information associated with each changed canonical set can itself be updated in constant time per change.
\end{theorem}
\begin{proof}
\ifFull
We associate with each canonical set an SDR sketch
%IBF 
for its elements that 
%is capable of performing the \textsf{listItems} operation on sets 
supports the \texttt{report} function for outputs of size up to $m$, and we
let $\oplus$ be defined by the operations of adding % and subtracting IBFs 
SDR sketches
as described in the previous section. The result of applying $\oplus$ to the 
%IBFs
SDR sketches
stored at each of the canonical sets  will be a single 
%IBF 
sketch for each of the queried ranges.
Then, to report the set-difference, we call the 
%\textsf{listItems} operation on this IBF, and
\texttt{report} function which
%report the items listed 
outputs the elements if there are fewer than $m$ of them. The SDR sketch is implemented as
an IBF, and the \texttt{report} function is supported via an IBF subtraction followed by
the \textsf{listItems} operation. See Section~\ref{app:IBF}
%won't be in appendix since we're in the full version, but had to fix labels
%quickly
for details of the implementation. 
If the
\textsf{listItems} operation fails to decode the IBF, or if it lists more than
$m$ items, we report that the cardinality bound is exceeded.
\else
See appendix \ref{app:proofs}.
\fi
\end{proof}
\begin{theorem}
\ifFull \label{thm:variable-cardinality-reportingFull} \else \label{thm:variable-cardinality-reporting} \fi{} 
Suppose that we wish to report set differences that may be large or small,
without the fixed bound $m$, in a time bound that depends on the size of the
difference but that may depend only weakly on the size of the total data set. In
this case, we can answer range difference queries with probability at least $1 - \epsilon$ 
for any range space that can
be modeled by a canonical group or semigroup range searching data structure. Our
solution stores a number of words of aggregate information per canonical set
that is $O(1)$ per element of the set, uses a combination operation $\oplus$
that takes time $O(m)$ (where $m$ is the cardinality of the final set-theoretic
difference) and allows the result of this combination to be decoded to yield the
query results in $O(m)$ time. If the data structure is updated, the aggregate
information associated with each changed canonical set can itself be updated in
logarithmic time per change.  
\end{theorem}
\begin{proof}
\ifFull
For each canonical set $S$ of the range query data structure, we store an SDR sketch
consisting of a hierarchy of $\log_2
|S|$ IBFs, capable of successfully performing the \textsf{listItems} operation
on sets of sizes that are powers of two up to the size of $S$ itself. The total
size of these IBFs adds in a geometric series to be proportional to the size of
$S$. The time to update the data structure then follows since each element of
each canonical set is stored in a logarithmic number of IBFs.

To answer a range query, suppose first that $m$ is known. In this case, we
choose $j=\lceil\log_2 d\rceil$, so that $2^j>m$ but $2^j=O(m)$, and we collect
and combine the IBFs of size $2^j$ exactly as in
Theorem~\ifFull{\ref{thm:fixed-cardinality-reportingFull}}\else{\ref{thm:fixed-cardinality-reporting}}\fi{}. The only possible complication is
that, for some canonical sets, the number $2^j$ may be larger than the number of
elements in the canonical set, so that we have no pre-stored IBF of the correct
size to collect. In this case, though, we can construct an IBF of the correct
size from the list of elements of the canonical set, in time $O(2^j)$.
Finally, if $m$ is unknown, we try the same procedure for successive powers of two until finding a set of powers for which the result succeeds.
\else
See Appendix \ref{app:proofs}.
\fi
\end{proof}
\begin{theorem}
\ifFull \label{thm:size-estimationFull} \else \label{thm:size-estimation} \fi{} 
Suppose that we wish to report the cardinality of the set difference rather than
its elements, and further that we allow this cardinality to be reported
approximately, within a fixed approximation ratio that may be arbitrarily close
to one.  In this case, we can answer range difference queries with 
probability at least $1 - \epsilon$ for any range
space that can be modeled by a canonical group or semigroup range searching data
structure.  Our solution stores a number of words of aggregate information per
canonical set that has size $O(\log n \log U)$, uses a combination operation $\oplus$
that takes time $O(\log n \log U)$ and allows the result of this combination to be
decoded to yield the query results in $O(\log n \log U)$ time. 
%\textbf{are these log n bounds right?}
\end{theorem}
\ifStrata
\begin{proof}
We associate a strata estimator with each canonical set, and we let $\oplus$ be
defined by the operations of adding and subtracting strata estimators as
described in the previous section. The result of applying $\oplus$ to these
strata estimators will be a single strata estimator for the queried range
difference, from which we may use the estimator to approximate the cardinality
of the difference. The space and time bounds are those for the strata estimator.
\end{proof}
\else
\begin{proof}
\ifFull
We associate with each canonical set an SDC sketch for its elements that supports the 
\texttt{count} function with probability $1 - \epsilon$.
and we let $\oplus$ be
defined by the operations of adding the SDC sketches.
The result of applying $\oplus$ to these
sketches will be a single sketch for each of the queried ranges.
Then to approximate the cardinality of the difference, we call the \texttt{count} function.
We implement each SDC sketch as set of $O(\log(1/\epsilon))$ independent linear 
frequency moment estimation sketches. 
The space and time bounds are the same as those for computing these sketches, and further
details of the implementation are given in Section~\ref{sec:fme}.
\else
See Appendix \ref{app:proofs}.
\fi
\end{proof}

%\subsection{Partial Sum Queries}
%In Figure~\ifFull{\ref{figPrefixSumsFull}}\else{\ref{figPrefixSums}}\fi{}, we illustrate a
%standard data structure for
%partial sum queries when subtraction is allowed.
%We give details for other range searching data structures in Appendix~\ref{app:}

%\begin{figure}[hbt!]
%\centering
%\includegraphics[scale=.4]{prefixSums.pdf}
%\caption{
%\label{figPrefixSums}%
%A data structure for
%partial sum queries when subtraction is allowed.
%The prefix sums 
%required to answer any partial sum query can be precomputed in constant time per
%array element. Given these precomputed canonical subsets, we can answer a
%partial sum query in constant time (assuming $d$ is constant). We illustrate
%this concept in the case where $d = 2$.
%The region $X$ is all the elements of $B$ with
%indices $<i,<j$. $Y$ is all elements of $B$ with indices  $<i, < k$. $Z$ is
%all elements with indices $<k, <j$. $Z$ is all elements with indices
%$\leq k, \leq \ell$. The query rectangle $r$ is all elements between
%indices $i,j$ and $k, \ell$ inclusive. Thus  $\sum w(r) = \sum w(W) - \sum
%w(Z) - \sum w(Y) + \sum w(X)$. }
%\end{figure}

\section{Invertible Bloom Filters}
\ifFull{\label{app:IBF}}\else{\label{sec:IBF}}\fi{}
%Central to our results is 
We implement our SDR sketches using the invertible Bloom filter (IBF)~\cite{eg-sirtd-11},
a variant of the Bloom filter~\cite{bloom1970space} for maintaining sets of
items that extends it in three crucial ways that are central to our application.
First, like the counting Bloom filter~\cite{bmpsv-iccbf-06,fcab-scswa-00}, the
IBF allows both insertions and deletions, and it allows the number of inserted
elements to far exceed the capacity of the data structure as long as most of the
inserted elements are deleted later. Second, unlike the counting Bloom filter,
the IBF allows the 
elements of the set to be listed back out. And third, again unlike the counting
Bloom filter, the IBF allows \emph{false deletions}, deletions of elements that
were never inserted, 
and again it allows the elements involved in false deletion operations to be
listed back out as long as their number is small.  These properties allow us to represent large sets as small
IBFs, and to quickly determine the 
elements in the symmetric difference of the sets, as we now detail.

\ifFull
We assume that the items being considered belong to a countable
universe, $U$, and that given any element, $x$, we can determine a
unique integer identifier for $x$ in $O(1)$ time (e.g., either by $x$'s
binary representation or through a hash function).
Thus, for the remainder of this paper, we assume w.l.o.g.~that each
element $x$ is an integer.

The main components of an IBF are a table, $T$, of $t$ cells, for a given
parameter, $t$; a set of
$k$ random hash functions, $h_1,h_2,\ldots,h_k$, which map any 
element $x$ to $k$ distinct cells in $T$;
and a random hash function, $g$, which maps any element $x$ to a random
value in the range $[1,2^\lambda]$, where $\lambda$ is a specified
number of bits.

Each cell of $T$ contains the following fields:
\begin{itemize}
\item
\textsf{count}: an integer count of the number of items mapped to
this cell
\item
\textsf{idSum}: a sum of all the items mapped to this
cell
\item
\textsf{gSum}: a sum of $g(x)$ values for all items mapped to this
cell.
\end{itemize}
The \textsf{gSum} field is used for checksum purposes.
An IBF supports several simple algorithms 
for item insertion, deletion, and membership queries,
as shown in Figure~\ref{fig:algsFull}.
\else
For the remainder of this paper, we assume without loss of generality that each
element $x$ is an integer.
An IBF supports several simple algorithms 
for item insertion, deletion, and membership queries;
for a review of these basic details of the IBF see 
\ifFull Figure~\ref{fig:algsFull}. \else Appendix \ref{app:IBF-code}. \fi
\fi{}

\ifFull
\begin{figure}[t]
\noindent
{\sf insert}$(x)$:
\begin{algorithmic}[100]
\FOR {each $h_i(x)$ value, for $i=1,\ldots,k$}
\STATE
add $1$ to $T[h_i(x)].\mbox{\textsf{count}}$
\STATE
add $x$ to $T[h_i(x)].\mbox{\textsf{idSum}}$
\STATE
add $g(x)$ to $T[h_i(x)].\mbox{\textsf{gSum}}$
\ENDFOR
\end{algorithmic}

\medskip
\noindent
{\sf delete}$(x)$:
\begin{algorithmic}[100]
\FOR {each $h_i(x)$ value, for $i=1,\ldots,k$}
\STATE
subtract $1$ from $T[h_i(x)].\mbox{\textsf{count}}$
\STATE
subtract $x$ from $T[h_i(x)].\mbox{\textsf{idSum}}$
\STATE
subtract $g(x)$ from $T[h_i(x)].\mbox{\textsf{gSum}}$
\ENDFOR
\end{algorithmic}

\medskip
\noindent
{\sf isMember}$(x)$:
\begin{algorithmic}[100]
\FOR {each $h_i(x)$ value, for $i=1,\ldots,k$}
\IF {$T[h_i(x)].\mbox{\textsf{count}}= 0$ \textbf{and}
      $T[h_i(x)].\mbox{\textsf{idSum}}= 0$  \textbf{and}
      $T[h_i(x)].\mbox{\textsf{gSum}}= 0$}
\STATE 
\textbf{return} \textbf{false}
\ELSIF {$T[h_i(x)].\mbox{\textsf{count}}= 1$ \textbf{and}
      $T[h_i(x)].\mbox{\textsf{idSum}}= x$  \textbf{and}
      $T[h_i(x)].\mbox{\textsf{gSum}}= g(x)$ }
\STATE 
\textbf{return} \textbf{true}
\ENDIF
\ENDFOR
\STATE \textbf{return} ``not determined''
\end{algorithmic}

\medskip
\noindent
{\sf subtract}$(A,B,C)$:
\begin{algorithmic}[100]
\FOR {$i=0$ to $t-1$}
\STATE $T_C[i].\mbox{\textsf{count}} = T_A[i].\mbox{\textsf{count}} - T_B[i].\mbox{\textsf{count}}$
\STATE $T_C[i].\mbox{\textsf{idSum}} = T_A[i].\mbox{\textsf{idSum}} - T_B[i].\mbox{\textsf{idSum}}$
\STATE $T_C[i].\mbox{\textsf{gSum}} = T_A[i].\mbox{\textsf{gSum}} - T_B[i].\mbox{\textsf{gSum}}$
\ENDFOR
\end{algorithmic}

\medskip
\noindent {\sf listItems}$()$:
\begin{algorithmic}[100]
\WHILE {there is an $i\in[1,t]$ such that 
$T[i].\mbox{\textsf{count}} = 1$ \textbf{or} $T[i].\mbox{\textsf{count}} = -1$}
\IF {$T[h_i(x)].\mbox{\textsf{count}}\,=\, 1$ 
     \textbf{and} $T[h_i(x)].\mbox{\textsf{gSum}}\,=\, 
	     g(T[h_i(x)].\mbox{\textsf{idSum}})$ }
\STATE add the item,
$(T[i].\mbox{\textsf{idSum}})$,
to the ``positive'' output list
\STATE call {\sf delete}($T[i].\mbox{\textsf{idSum}}$)
\ELSIF {$T[h_i(x)].\mbox{\textsf{count}}\,=\, -1$ 
       \textbf{and} $-T[h_i(x)].\mbox{\textsf{gSum}}\,=\, 
       g(-T[h_i(x)].\mbox{\textsf{idSum}})$ }
\STATE add the item,
$(-T[i].\mbox{\textsf{idSum}})$,
to the ``negative'' output list
\STATE call {\sf insert}($-T[i].\mbox{\textsf{idSum}}$)
\ENDIF
\ENDWHILE
\end{algorithmic}

\caption{Operations supported by an invertible Bloom filter.}
\ifFull \label{fig:algsFull} \else \label{fig:algs} \fi{} 
\end{figure}
\fi{}

In addition, we can take the difference of one IBF, $A$, with a table $T_A$, and
another one, $B$, with table $T_B$, to produce an IBF, $C$, with table $T_C$,
representing their signed difference, with the items in $A\setminus B$ having
positive signs for their cell fields in $C$ and items in $B\setminus A$ having
negative signs for their cell fields in $C$ (we assume that $C$ is initially
empty).  This simple method is also shown in
\ifFull Figure~\ref{fig:algsFull}. \else Appendix \ref{app:IBF-code}. \fi

Finally, given an IBF, which may have been produced either through
insertions and deletions or through a subtract operation,
we can list out its contents by repeatedly looking for cells with
counts of $+1$ or $-1$ and removing the items for those cells if they
pass a test for consistency.
This method therefore produces a list of items that had positive
signs and a list of items that had negative signs, and 
is shown in \ifFull Figure~\ref{fig:algsFull}. \else Appendix
\ref{app:IBF-code}. \fi
In the case of an IBF, $C$, that is the result of a \textsf{subtract}($A,B,C$)
operation, the positive-signed elements belong to $A\setminus B$ and
the negative-signed elements belong to $B\setminus A$.

\ifStrata
\subsection{The Strata Estimator Algorithm}
The strata estimator~\cite{eguv-esswp-11} combines invertible Bloom filters with
a hierarchical sampling technique 
\ifFull
also used in a different context by Flajolet
and Martin~\cite{Flajolet} and Cormode \textit{et al.}~\cite{cmr-smidd-05}.
\fi{}
The
strata estimator for a set $S$, drawn from a universe of size $U$, consists of
an array of $\log_2 U$ IBF's, numbered $0$, $1$, $2$, $ \dots$, each having the
same number of cells $k$, and a hash function $h$ (independent of the hash
functions used in each IBF).

Specifically, to construct the strata estimator for a set $S$, for each element
$x$ of $S$, we compute $h (x)$, let $\ell$ be the number of trailing zeros in
$h(x)$, and place $x$ into IBF $\ell$, as shown in the pseudocode in
Figure~\ref{fig:strataFull}. Thus, the expected cardinality of each layer $i$ is
$1/2^{i+1}$ times the cardinality of $S$.

To estimate the Hamming distance of two sets, we subtract the corresponding IBFs
in each layer, and find the smallest-numbered layer $\ell$ for which we can
completely decode the difference between the two IBFs at layer $\ell$ and all
higher layers. If all layers are decodable, we have determined the exact
difference; otherwise, we multiply the number of differing elements in layers
$\ell$ and above by $2^\ell $, as in Figure~\ref{fig:strataFull}.

%added minipages to layout in two columns and compress space
\ifFull  %%% should we keep this figure in?
\begin{figure}[ht]% Placing it "h" makes more sense logically but takes more room
%\hspace{.1in}
%\begin{minipage}[t]{.42\textwidth}
\noindent{\sf strata}$(S)$:
  \begin{algorithmic}
    \FOR{each element $x$ in $S$}
    \STATE $y = h(x)$; $\ell = \log_2 (y \,\,\&\,\,\sim\! (y - 1))$
    \STATE add $x$ to the IBF for layer $\ell$
    \ENDFOR
  \end{algorithmic}
%\end{minipage}
%\hspace{.2in}
%\begin{minipage}[t]{.58\textwidth}
\medskip
\noindent{\sf estimate-distance}$(S,T)$:
\begin{algorithmic}
\STATE allocate storage for an IBF $C$
\STATE $d=0$
\FOR{$i=\log_2 U, \log_2 U -1,\dots, 2,1, 0$}
\STATE {\sf subtract}(layer $i$ of $S$, layer $i$ of $T$, $C$)
\IF{$C$ can be completely decoded}
\STATE add the number of decoded elements to $d$
\ELSE
\STATE return $2^{i+1}\times d$
\ENDIF
\ENDFOR 
\STATE return $d$
\end{algorithmic}
%\end{minipage}
\caption{Constructing a strata estimator, and estimating distances from two strata estimators.}
\ifFull \label{fig:strataFull} \else \label{fig:strata} \fi{} 
\end{figure}
\fi{}
\else

\ifFull
\begin{remark}
The \emph{Strata Estimator}~\cite{eguv-esswp-11} combines invertible Bloom filters with
a hierarchical sampling technique.
The strata estimator 
can also be used to implement the SDC sketch. 
However, compared to the frequency moment estimation technique
outlined below, the strata estimator has a slightly worse dependence
on $\epsilon$ and $\delta$, and therefore we do not incorporate it as a
component in our results.
%
%for a set $S$, drawn from a universe of size $U$, consists of
%an array of $\log_2 U$ IBF's, numbered $0$, $1$, $2$, $ \dots$, each having the
%same number of cells $k$, and a hash function $h$ (independent of the hash
%functions used in each IBF).
%
%Specifically, to construct the strata estimator for a set $S$, for each element
%$x$ of $S$, we compute $h (x)$, let $\ell$ be the number of trailing zeros in
%$h(x)$, and place $x$ into IBF $\ell$, as shown in the pseudocode in
%Figure~\ref{fig:strataFull}. Thus, the expected cardinality of each layer $i$ is
%$1/2^{i+1}$ times the cardinality of $S$.
%
%To estimate the Hamming distance of two sets, we subtract the corresponding IBFs
%in each layer, and find the smallest-numbered layer $\ell$ for which we can
%completely decode the difference between the two IBFs at layer $\ell$ and all
%higher layers. If all layers are decodable, we have determined the exact
%difference; otherwise, we multiply the number of differing elements in layers
%$\ell$ and above by $2^\ell $, as in Figure~\ref{fig:strataFull}.
\end{remark}
\fi

\subsection{Analysis}
\ifFull{\label{app:analysis}}\else{\label{sec:analysis}}\fi{}
In this section, we extend previous analyses~\cite{eg-sirtd-11,eguv-esswp-11} to
bound the failure probability for the functioning of an invertible Bloom filter
\ifStrata
and the strata estimator
\fi
to be less than a given parameter, $\epsilon>0$, which
need not be a constant (e.g., its value could be a function of other
parameters).

\begin{theorem}
\ifFull \label{thm:bloom-decodeFull} \else \label{thm:bloom-decode} \fi{} 
Suppose $X$ and $Y$ are sets with $m$ elements in their
symmetric difference, i.e., $m=|X\bigtriangleup Y|$, and let
$\epsilon>0$ be an arbitrary real number.
Let $A$ and $B$ be invertible Bloom filters
built from $X$ and $Y$, respectively, such that each IBF has
$\lambda\ge k+\lceil\log k\rceil$ bits in its \textsf{gSum} field, i.e.,
the range of $g$ is $[1,2^\lambda)$,
and each IBF has at least
$2km$ cells, where $k > \lceil \log (m/\epsilon)\rceil+1$ is the
number of hash functions used.
Then the \textsf{listItems} method for the IBF $C$ resulting from the
\textsf{subtract}$(A,B,C)$ method will list all $m$
elements of $X\bigtriangleup Y$ and identify which belong
to $X\setminus Y$ and which belong to $Y\setminus X$ with 
probability at least $1-\epsilon$.
\end{theorem}
\begin{proof}
\ifFull
There are at least $km$ empty cells in $C$; hence,
the probability that an element 
$z\in X\bigtriangleup Y$ collides with another element
in $X\bigtriangleup Y$ in each of its $k$ cells
is at most $2^{-k}$, independent of 
the locations chosen in $C$ by the other 
elements in $X\bigtriangleup Y$.
Thus, the probability that any element in $X\bigtriangleup Y$ has $k$
collisions is at most $m\, 2^{-k} \le \epsilon/2$.
Therefore, with probability at least $1-\epsilon/2$, every 
element in $X\bigtriangleup Y$ has at least one cell in $C$ where it is
the only occupant.

Next, consider the probability that a cell has
$|\textsf{count}|=1$ and $|\textsf{gSum}| = g(|\textsf{idSum}|)$ but
nevertheless has two or more occupants.
Since the \textsf{gSum} field has
at least $k$ bits,
the probability that a cell is ``bad'' in this way
is at most
$2^{-\lambda}$. 
Thus, since there are at most $km$ non-zero cells,
the probability we have any such bad cell in $C$ is at
most
\[
km\, 2^{-\lambda} \le m 2^{-k} \le \epsilon/2.
\]
This completes the proof.
\else
See Appendix \ref{app:proofs}.
\fi
\end{proof}

%\ifFull
%%% A corollary in terms of U doesn't make sense
%We also have the following.
%
%\fi{}
%\begin{corollary}
%\ifFull \label{cor:bloom-decodeFull} \else \label{cor:bloom-decode} \fi{} 
%Let $X$ and $Y$ be two sets
%and let $m$ be the size of their symmetric difference.
%Then, using IBFs $A$ and $B$ defined with $O(\log U)$ hash functions and
%$O(m\log U)$ cells, and with $O(\log U)$ bits per cell,
%we can decode $X\setminus Y$ and $Y\setminus X$ from the difference
%of $A$ and $B$ with probability at least $1-U^{-c}$, for any constant
%$c>0$.
%\end{corollary}
%\begin{proof}
%\ifFull
%{
%For the success probability,
%apply 
%Theorem~\ifFull{\ref{thm:bloom-decodeFull}}\else{\ref{thm:bloom-decode}}\fi{}  with $\epsilon=U^{-c}$. 
%To achieve the word-size bound, note that it is sufficient
%to do all arithmetic modulo $2^\lambda$ and use identifiers that are
%numbered $1$ to $U$.
%}
%\else
%{
%    See Corollary \ref{cor:bloom-decodeFull} in Appendix \ref{app:proofs}.
%}
%\fi{}
%\end{proof}

To avoid infinite loops, we make a small change to \textsf{listItems}, forcing
it to stop decoding after $m$ items have been decoded regardless of whether
there remain any decodable cells. This change does not affect the failure probability and with it the running time is always $O(mk)$.

\ifStrata
Let us next consider the accuracy of our strata estimator scheme.
Suppose the symmetric difference
between two sets $X$ and $Y$, has cardinality $m$, and we have
constructed a strata estimator $S$ as the difference between a strata
hierarchy for
$X$ and a strata hierarchy for $Y$ (using the same hash functions and
size parameters).
Then the $i$th stratum in $S$ has expected size $m/2^i$. 
Define the \emph{cardinality} of a stratum in $S$ to be the number of
elements in the symmetric difference of $X$ and $Y$
that are mapped there by the sampling
process (independent of whether or not
we can successfully decode them in that level's IBF).
\ifFull
\begin{lemma}
\ifFull \label{lem:strata-sizesFull} \else \label{lem:strata-sizes} \fi{} 
For any $\epsilon>0$, with probability $1-\epsilon$, the
cardinality of each of the strata from $0$ to $i$ in $S$
is within a multiplicative factor of
$(1\pm\delta)$ of its expectation, for
$\delta \ge \sqrt{4(2+\log (1/\epsilon))(2^i/m)\ln 2}$.
\end{lemma}
\begin{proof}
\ifFull
Let $Z_j$ be a random variable representing the size of the $j$th stratum,
for $j=0,1,\ldots,i$,
and let 
$\mu_j = E(Z_j) = m/2^j$.
By standard Chernoff bounds, 
\ifFull
% full version is in appendix, but citation still affecting length
%(e.g., see~\cite{mu-pcrap-05,mr-ra-95}),
\fi{}
for $\delta > 0$,
%using align instead compresses space,
% but why not also use it in full version
%\ifFull
%\[
%  \Pr\left(Z_j > (1+\delta)\mu_j\right) \,< \,e^{-\mu_j\delta^2/4} 
%\]
%and
%\[
%  \Pr\left(Z_j < (1-\delta)\mu_j\right)  \,<\, e^{-\mu_j\delta^2/4}.
%\]
%\else
\begin{align} 
  \Pr\left(Z_j > (1+\delta)\mu_j\right) & < e^{-\mu_j\delta^2/4} \\
  \mbox{and \quad} 
  Pr\left(Z_j < (1-\delta)\mu_j\right) & < e^{-\mu_j\delta^2/4}.
\end{align}
%\fi{}
By taking $\delta \ge \sqrt{4(2+\log (1/\epsilon))(2^i/m)\ln 2}$, 
we have that the probability that a particular $Z_j$ is
not within a multiplicative factor of
$1\pm \delta$ of its expectation is at most
\[
2e^{-\mu_j\delta^2/4} \,\le\, 2^{-(1+\log (1/\epsilon))2^{i-j}} 
   \, =\, \frac{(\epsilon/{2})}{2^{2^{i-j}}}.
\]
Thus, by a union bound, the probability that any $Z_j$
is not within a multiplicative factor of $1\pm \delta$ of its
expectation is at most
\[
\sum_{j=0}^i \frac{(\epsilon/{2})}{2^{2^{i-j}}}
\,
= 
\,
({\epsilon}/{2})  \sum_{j=0}^i 2^{-2^{i-j}} ,
\]
which is at most
$\epsilon$.
\else
{
    See Lemma \ref{lem:strata-sizesFull} in Appendix \ref{app:proofs}. 
}
\fi{}
\end{proof}

Combining the above results,
we have the following:
\fi{}
\begin{theorem}
\ifFull \label{thm:strataFull} \else \label{thm:strata} \fi{} 
Let $X$ and $Y$ whose symmetric difference has size $m$,
and let $0<\epsilon<1$ and $0<\delta<1$ be arbitrary real numbers.
Suppose we encode $X$ and $Y$ in a strata estimator, $S$, of height
$\lceil\log U\rceil$, 
with each member IBF having $2km'$ cells, $k$ hash functions, and 
\textsf{gSum} size of $k+\lceil\log k\rceil$ bits, for 
$k\ge \lceil\log 2m'/\epsilon\rceil+1$ and 
$m' \ge 4\delta^{-2}(2+\log (2/\epsilon))\ln 2$,
which is $\Omega(\delta^{-2}\log (1/\epsilon))$. 
Then, with probability at least
$1-\epsilon$, the estimate resulting from our decoding algorithm,
run on the first IBF that is not too full to decode, will
be within a factor of $1\pm\delta$ of $m$.
\end{theorem}

\begin{proof}
\ifFull
{
By Lemma~\ifFull{\ref{lem:strata-sizesFull}}\else{\ref{lem:strata-sizes}}\fi{},
for $\delta = \sqrt{4(2+\log (2/\epsilon))(2^i/m)\ln 2}$, 
the cardinalities of the first $i$ levels of a strata
estimator are each within a factor of $1\pm\delta$ 
of its expectation, $m/2^i$, with probability $1-\epsilon/2$.
We don't know the value of $m$ when we construct the
strata estimators, but, by Theorem~\ifFull{\ref{thm:bloom-decodeFull}}\else{\ref{thm:bloom-decode}}\fi{}, 
if each IBF has $2km'$ cells, 
where $k = \lceil \log 2m'/\epsilon\rceil+1$
and $m' \ge 2$, then we can decode a set of
$m'$ elements with probability at least $1-\epsilon/2$.
Thus, we can assume without loss of generality that $m>m'$, for otherwise with high probability we will 
decode the entire strata estimator and determine the correct distance without error.
With probability $1-\epsilon$,
we can achieve an approximation factor of $1\pm\delta$ for our size
estimator by taking
\ifFull
\[
m' \ge 4\delta^{-2}(2+\log (2/\epsilon))\ln 2. 
\]
\else
$m' \ge 4\delta^{-2}(2+\log (2/\epsilon))\ln 2. $
\fi{}
Depending on the level, $i$, where we can first successfully decode
the contents of its IBF,
this gives us a $(1\pm\delta)$ estimate for $m/2^i$, which we convert
into a $(1\pm\delta)$ estimate for $m$ by multiplying by $2^i$.
}
\else
{
    See Theorem \ref{thm:strata} Appendix \ref{app:proofs}. 
}
\fi{}
\end{proof}
\ifFull
Thus, we have the following.
\fi{}
\begin{corollary}
\ifFull \label{cor:strataFull} \else \label{cor:strata} \fi{} 
Let $0<\epsilon<1$ and $0<\delta<1$ be arbitrary real numbers.
Given two sets, $X$ and $Y$, 
taken from a universe of size $U$, 
we can compute an estimate $\hat m$ such that
\[
(1-\delta)|X\bigtriangleup Y|\le {\hat m}\le
(1+\delta)|X\bigtriangleup Y|,
\]
with probability $1-\epsilon$, using
a strata estimator of height $\lceil \log U\rceil$
with each member IBF being of size
$t=\Theta(\delta^{-2}(\log (1/\epsilon))(\log (1/\delta) + \log (1/\epsilon))$
and having
$k=\Theta(\log (1/\delta) + \log (1/\epsilon))$ hash functions. 
The preprocessing time is $O(k|X| + k|Y| + t\log U)$ 
and the time to compute the estimate $\hat m$ is
$O(t\log U)$.
\end{corollary}
\fi

\ifFull
For example, we can achieve an approximation factor of $1\pm o(1)$
with probability $1-U^{-c}$, for any constant $c>0$, with a strata
estimator that has size polylogarithmic in $U$.

\begin{remark}
Suppose we only need to answer the decision version of the problem: 
``Is $|X\bigtriangleup Y| (1\pm \delta)\leq \tau$?'' for some threshold
parameter $\tau$. 
Then we can save a $\log U$ factor in both time and space by only constructing a single
layer of the strata estimator, chosen according to $\tau$. 
\end{remark}
\fi{}

\ifFull
\section{Measuring and Estimating Set Dissimilarity}
\label{app:diss}
Many distances have been defined between pairs of sets, and different choices 
may be appropriate in different applications. For instance, in information retrieval, it 
is useful 
to normalize the distance between sets by weighting distance
by the size of the sets, so that documents on similar topics may 
count as similar even if their lengths may greatly differ. 
In applications where the sets being compared describe the presence or 
absence of a fixed collection of binary attributes, 
Hamming distance may be more appropriate.

In more detail, suppose we are given finite sets $X$ and $Y$ either by listings
of their elements or by bit vectors $\vec X$ and $\vec Y$ that represent their
characteristic functions as subsets of a fixed 
universe~$U$.
There are a host of ways to define notions of distance between $X$
and $Y$, with the following being some of the most well-known:
\begin{itemize}
\item
\emph{Set-theoretic Hamming distance}: 
$H(X,Y)=|X\bigtriangleup Y|=|X\setminus Y|+|Y\setminus X|$,
that is, the number of positions where $\vec X$ and $\vec Y$ differ.
\item
\emph{Jaccard/Tanimoto distance}: $J(X,Y)=|X\bigtriangleup Y|/|X\cup
Y|=1-j(X,Y)$, where $j(X,Y)$ is the 
related \emph{Jaccard index}~\cite{Hamers1989315}, defined as $j(X,Y)=|X\cap
Y|/|X\cup Y|$. The \emph{Tanimoto coefficient}~\cite{l-ptitd-99},
$t(\vec X,\vec Y)= (\vec X \cdot \vec Y)/(||\vec X||^2+||\vec Y||^2 - \vec
X\cdot \vec Y)$ generalizes this distance to vectors.
\item
\emph{Dice dissimilarity}: $D(X,Y)=|X\bigtriangleup Y|/(|X|+|Y|)=1-\gamma(X,Y)$,
where $\gamma(X,Y)$ is 
\emph{Dice's coefficient} 
$\gamma(X,Y)=2|X\cap Y|/(|X|+|Y|)$~\cite{d-maeas-45}.
\item
\emph{Tversky dissimilarity}: $T(X,Y)=1-t(X,Y)$, where $t(X,Y)$ is the
\emph{Tversky index}
$t(X, Y) = |X \cap Y|/(|X\cap Y| + \alpha|X\setminus Y|+ \beta|Y\setminus X|)$
parameterized by $\alpha,\beta\ge 0$~\cite{Tversky1977327}.
For instance, setting $\alpha=\beta=1$ is equivalent to the
Jaccard/Tanimoto index and setting $\alpha=\beta=1/2$ is
equivalent to Dice's coefficient.
Here, we are interested in any constant non-zero 
setting of $\alpha$ and $\beta$.
\end{itemize}
\ifFull
For a detailed discussion of these metrics, indexes, and
coefficients, as well as others, 
please see~\cite{lg-sm-07,sb-bafes-07}.

The Hamming and Jaccard/Tanimoto distances are metrics~\cite{bh-aiist-09,lg-sm-07,l-ptitd-99}, 
whereas Dice dissimilarity
is a semimetric that doesn't always satisfy the
triangle inequality~\cite{c-setra-02,lg-sm-07}.
Depending on $\alpha$ and $\beta$,
Tversky dissimilarity could be
a metric, semimetric, quasimetric, or quasisemimetric, in that it may
or may not satisfy the triangle inequality and/or the symmetry axiom. However,
\else
The Hamming and Jaccard/Tanimoto distances are metrics, whereas the Dice and Tversky 
dissimilarities need not always satisfy the triangle inequality; however,
\fi{}
this difference is not significant for our distance estimation algorithms.

\subsection{Estimates for Other Dissimilarities}
Given $|A|$ and $|B|$ and an accurate estimate for $H(A,B)$,
we can also estimate the cardinalities of the intersection and union
using the identities
\begin{eqnarray*}
|A\cap B| = (|A| +|B|-H(A,B))/2\\
|A\cup B| = (|A| +|B| + H(A,B))/2
\end{eqnarray*}

The estimate for the cardinality of the union preserves the multiplicative error
of the estimate for the Hamming distance. We cannot bound the multiplicative
error  of the estimated intersection cardinality, however, unless it is known to
be at least proportional to the Hamming distance.  

Once we have an accurate estimate for the Hamming distance $H(A,B)$ and the
cardinality of the union $A\cup B$, we can compute similarly accurate estimates
for each of the set-dissimilarity
metrics and semimetrics mentioned in the introduction.
\begin{itemize}
\item 
For the Dice dissimilarity, estimate $D(A,B)$ as $H(A,B)/(|A|+|B|)$. Since
$|A|$ and $|B|$ are exact, the accuracy does not change compared to that for
$H$.
\item 
For the Jaccard distance, estimate $J(A,B)$ as $2H(A,B)/(|A|+|B|+H(A,B))$. The
resulting estimate has a multiplicative error that is at least as good as the
error in estimating $H$.
\item 
For the Tversky dissimilarity,  assume without loss of generality that (as
estimated) $|A\setminus B| > |B\setminus A|$ (else we apply a symmetric formula
with the roles of $A$ and $B$ reversed), let $p$ be our estimate of the ratio
$|B\setminus A|/|A\setminus B|$, and let $q$ be our estimate of the ratio
$|A\cap B|/|A\setminus B|$. We then estimate $T(A,B)=(\alpha+\beta
p)/(\alpha+\beta p+1)$. Both $p$ and $q$ are estimated with additive error
$O(\epsilon)$, and since both the numerator and denominator of our estimate of
$T(A,B)$ are $\Omega(1)$ and are the sum of additively-accurate estimators, they
are also estimated to within multiplicative error $1+O(\epsilon)$. Therefore,
their ratio also has multiplicative error $1+O(\epsilon)$.  
\end{itemize}
\fi{}

%%%%   SODA doesn't want appendices, instead, wants full version inline
%%%%
%\ifMapReduce %no room in combined paper?
%\section{A MapReduce Implementation}
%In this section, we establish that our scheme can be implemented in
%the MapReduce framework.
%
%\begin{theorem}
%\ifFull \label{thm:mapreduceFull} \else \label{thm:mapreduce} \fi{} 
%Suppose we are
%given a collection of sets, $X_1,\ldots,X_n$, taken from a universe
%of size $U$,
%and let $M$ denote an upper bound on the size of a reducer's
%input/output in any round.
%Given a method for mapping likely pairs of sets with small
%dissimilarity in $O(\log_M U)$ MapReduce rounds, we can compute 
%an estimate of the size of the symmetric difference
%for each such pair of sets, with accuracy $1\pm\delta$ with
%probability $1-\epsilon$, for $0<\delta,\epsilon<1$,
%in $O(\log_M U)$ MapReduce steps.
%\end{theorem}
%
%For the proof, see Appendix~\ifFull{\ref{app:mapreduceFull}}\else{\ref{app:mapreduce}}\fi{}.
%As a consequence, for example, if $M$ is $O(U^{1/d})$, for some constant $d\ge1$,
%then our MapReduce algorithm runs in a constant number of rounds.
%\else {}
%\fi{}
\ifMapReduce
\section{A MapReduce Implementation}
\label{app:mapreduce}
Suppose we are
given a collection of sets, $X_1,\ldots,X_n$, taken from a universe
of size $U$,
and let $M$ denote an upper bound on the size of a reducer's
input/output in any round.
Given a method for mapping likely pairs of sets with small
dissimilarity in $O(\log_M U)$ MapReduce rounds, we can compute 
an estimate of the size of the symmetric difference
for each such pair of sets, with accuracy $1\pm\delta$ with
probability $1-\epsilon$, for $0<\delta,\epsilon<1$,
in $O(\log_M U)$ MapReduce steps.
The proof of this fact is as follows.

\begin{proof}
In the MapReduce framework (e.g., see~\cite{dg-msdpl-08,arxiv1004.4708}),
a computation is specified as a sequence of
map, shuffle, and reduce steps that operate on a set
$X=\{x_1,x_2,\ldots,x_p\}$ of values:
\begin{itemize}
\item
A {\em map step} applies a function, $\mu$, to each value, $x_i$, to
produce a finite set of key-value pairs $(k,v)$. To allow for
parallel
execution, the computation of the function $\mu(x_i)$ must depend
only on
$x_i$.

\item 
A {\em shuffle step} collects all the key-value pairs produced in the
previous map step, and produces a set of lists,
$L_k=(k;v_1,v_2,\ldots)$,
where each such list consists of all the values, $v_j$, such that
$k_j=k$
for a key $k$ assigned in the map step.

\item
A {\em reduce step} applies a function, $\rho$, to each list
$L_k=(k;v_1,v_2,\ldots)$, formed in the shuffle step, to produce a
set of
values, $y_1,y_2,\ldots\,$. The reduction function, $\rho$, is
allowed to
be defined sequentially on $L_k$, but should be independent of other
lists
$L_{k'}$ where $k'\not=k$.
\end{itemize}

\ifMapReduce
In Theorem~\ref{thm:mapreduce}, 
we follow the I/O-bound approach~\cite{arxiv1004.4708},
where we define $M$ to be an upper bound on the I/O-buffer memory
size for reducers used in a given MapReduce algorithm.
\fi

\paragraph{Preprocessing phase.}
In the preprocessing phase, we map each set $X$ to a hierarchy of
IBFs in parallel. For each such IBF, we need to compute the 
\textsf{count}, \textsf{idSum}, and \textsf{gSum} values for each cell.
We do this by first performing a map operation for each $x$ in $X$,
which maps $x$ to $k$ triples $(h_i(x), x, g(x))$. There might be too
many of these to directly send to one reducer, however, so we use the
\emph{invisible funnel} technique 
of Goodrich \textit{et al.}~\cite{arxiv1101.1902} to map these values to
a reducer for $h_i(x)$ in batches of size $O(M)$, which computes
intermediate sums for
the \textsf{count}, \textsf{idSum}, and \textsf{gSum} values.
After $O(\log_M |X|)$ rounds, we will have computed all these values
for the IBF for $X$.

\paragraph{Estimation step.}
Once we have a stratified IBF for each set $X$ in our collection, we
map these representations 
according to one of the previous heuristics for collecting
sets likely to be close in their dissimilarity distance
(e.g., see~\cite{vcl-epssj-10,wwlw-mdndo-10}).
We assume this step takes at most $O(\log_M U)$ steps, where $U$ is
the size of the universe of set values.

For each such pair of strata estimators, $S_X$ and $S_Y$, we map this
pair to a reducer, which then computes the estimation for the distance
between $X$ and $Y$ according to our algorithm, given above.
Of course, this assumes that $M$ is at least the size of a strata
estimator, which we feel is a reasonable assumption.
This completes the proof.
\end{proof}
\fi{}

\ifAllPairs{
% our current results aren't better than existing techniques.
% refs and sections related to all pairs have been moved to all-pairs.tex
% in case we want to revise them later.
%\input{all-pairs} 
}\fi{}

\section{Frequency Moment Estimation}
\label{sec:fme}
We implement our SDC sketches using frequency moment estimation techniques.
Let $x$ be a vector of length $U$, and suppose we have a data stream of length $m$, 
consisting of a sequence of updates to $x$ of the form
$(i_1, v_1), \ldots, (i_m, v_m) \in [U] \times [-M,M]$ for some $M > 0$.  
That is, each update is a pair $(i,v)$ which updates the $i$th coordinate of $x$ such that  $x_i \rightarrow x_i + v$.  

The frequency moment of a data stream is given by
\[
F_p = \sum_{i\in [U]} x_i^p = \|x\|_p^p.
\]

Since the seminal paper by Alon~\emph{et al.}~\cite{ams-scafm-99}, frequency moment
estimation has been an area of significant research interest. Indeed the full literature
on the subject is too rich to survey here. Instead, see e.g. the recent work by 
Kane~\emph{et al.}~\cite{knpw-fmedsop-11} and the references therein.
Kane~\emph{et al.}~\cite{knpw-fmedsop-11} gave algorithms for estimating $F_p$, $p \in (0,2)$.
Their algorithm requires $O(\log^2(1/\delta))$ time per update and  
$O(\delta^{-2} \log (mM))$ space.
However, faster results are known for estimating the second frequency moment $F_2$ with
constant probability.  
Thorup and Zhang~\cite{tz-tb4uh-04} 
and 
Charikar\emph{et al.}~\cite{ccf-ffids-02}
independently improve upon the original result of 
Alon~\emph{et al.}~\cite{ams-scafm-99}, 
 to achieve an optimal $O(1)$
update time using $O(\delta^{-2} \log (mM))$ space.

Given a sparse vector $X$ of length $m$ with coordinates bounded by $[-M,M]$,  
we can estimate $\|X\|_2^2$ by treating it as a data
stream and running the algorithm of Thorup and Zhang.
The algorithm computes a \emph{sketch} $S_X$ of $X$,
of size $O(\delta^{-2}\log{mM})$ such that $\|S_X\|_2^2$ is within a factor of
$O(1 \pm \delta)$ of $\|X\|_2^2$ with constant probability.
We can improve the probability bound to any arbitrary $O(1 -
\epsilon)$ by running the algorithm $O(\log (1/\epsilon))$ times independently to produce
$O(\log(1/\epsilon))$ independent sketches $S_{Xi}$, and taking the
median of $\|S_{Xi}\|_p^p.$ This strategy takes $O(\log(1/\epsilon))$ time to process each
non-zero element in $X$, and the space required is $O(\delta^{-2}\log(1/\epsilon) \log(mM))$.

Furthermore, each sketch is linear, and therefore we can estimate the frequency moment
of the difference of two sparse vectors $X$ and $Y$ by subtracting their sketches;
$\|S_X - S_Y\|_2^2$ is within a $O(1 \pm \delta)$ factor of $\|X-Y\|_2^2$ with
constant probability, and we can maintain $O(\log (1/\epsilon))$ independent
sketches to achieve probability $O(1 - \epsilon)$.

Now, suppose we want to estimate the Hamming distance between two sets.
We treat each set as a sparse bit-vector and use the fact that 
%(\textbf{TODO:} something about $\ell_0^2$ vs $\ell_2^2)$
the Hamming distance between two bit vectors is equivalent to the squared
Euclidean distance, 
which is just the second frequency moment of the difference of the vectors. 
 Then we can apply the above strategy for sparse
vectors to produce $O(\log (1/\epsilon))$ sketches for each set in 
$O(\log(1/\epsilon))$ time per element, and we
subtract all $O(\log 1/\epsilon)$ pairs of sketches in $O(\delta^{-2}\log (1/\epsilon)
\log U)$ time. Finally we compute the second frequency moment of each sketch and
take the median over all estimations in $O(\log 1/\epsilon)$ time.

%Furthermore, we can apply this strategy to an entire family of sets. 
We summarize this strategy in the following theorem.

\begin{theorem}
Let $0<\epsilon<1$ and $0<\delta<1$ be arbitrary real numbers.
Given two sets, $X$ and $Y$, 
taken from a universe of size $U$, 
we can compute an estimate $\hat m$ such that
\[
(1-\delta)|X\bigtriangleup Y|\le {\hat m}\le
(1+\delta)|X\bigtriangleup Y|,
\]
with probability $1-\epsilon$, using a sketch of size 
$O(\delta^{-2}\log (1/\epsilon) \log U)$
The preprocessing time, including the time required to initialize the sketch is 
$O(\delta^{-2}\log (1/\epsilon) \log U + (|Y| + |X|)\log(1/\epsilon))$ 
and the time to compute the estimate $\hat m$ is
$O(\delta^{-2}\log(1/\epsilon)\log U)$.
\end{theorem}

\ifFull
%We did the general framework, now lets show it off for specific examples
%\section{Set-Difference Range Searching}
\section{Set-Difference Range Searching Details}
\label{sec:sdDetails}
%Many data structures for designed to answer geometric searching queries 
%are based on a decomposition scheme.
In this section, we look at several examples of canonical semigroup range searching
structures and show how to transform them into data structures that can
efficiently answer set-difference range queries.  
Specifically, we give more detail for the results which were summarized in 
Table~\ref{tab:summaryFull}. 

%
%NOTE: do to our notation earlier, there isn't really a distinction between
%intersection and range searching
%
%We examine two special cases of geometric searching problems. 
%First we examine data structures for \emph{range searching}
%in which the result of a query 
%Then we examine a data structure for \emph{intersection searching} in which 
%$\diamond(\gamma,S)$ is the subset of $S$ which intersects $\gamma$.
%Regardless of the specific definition of $\diamond$, the data structures
%are based on decomposition schemes. 
%Therefore we can apply Theorem \ref{thmTransform} to  
%transform them into a data structures for 
%set-difference range searching or set-difference intersection searching.  
%
\subsection{Orthogonal Range Searching}
In the orthogonal range searching problem, our task is to preprocess a set of points in $R^d$
 such that given a query range defined by an axis-parallel hyper rectangle, we can efficiently report or count the set of points contained in the hyper-rectangle
The range tree is a classic canonical semigroup range searching data structure that
solves this problem.

In one dimension, a range tree is a binary tree in which each leaf represents a
small range containing a single point, and an in-order traversal of the leaves
corresponds to the sorted order of those points. The range of any interior node
is the union of the
ranges of its descendant leaves. 
Each node stores the value of the aggregate function on its 
canonical set---the
set of points in its range.

%Thus the range of the root node contains all the points stored in the range tree. 
%In a range tree for the orthogonal range
%counting problem, each node stores the sum of the weights of the points in its
%range. 

We can easily build a $d \geq 2$ dimensional range tree by building a range tree
over the $d^{\textnormal{th}}$ dimension, and at each internal node we store a
($d-1$)-dimensional dimensional range tree over the points in that node's range.
To answer a query using a range tree, we find a set of $O(\log^d n)$ canonical
ranges which exactly cover the query range and report the sum of the 
aggregate functions of each of these canonical ranges. 

Thus  
we can apply Theorem \ref{thm:fixed-cardinality-reportingFull},
\ref{thm:variable-cardinality-reportingFull}, or \ref{thm:size-estimationFull} 
 to obtain a data structure for set-difference orthogonal range queries.
We preprocess each canonical set into an appropriate sketch, and instead of storing the
value of the aggregate function in each node, instead we store a pointer to the
corresponding sketch for that node's canonical set. When answering a query, instead of
summing the aggregate functions, we simply add or subtract the sketches.

%Note that although it is common to apply fractional cascading techniques to gain
%a $O(\log n)$ speedup, we cannot apply the above transformation in conjunction
%with fractional cascading techniques. However, we do gain substantial savings in
%both space and query time when the size of the set difference is expected to be
%small. 
%
\subsection{Simplex Range Searching}
Partition trees~\cite{Matousek1992} 
are a linear space data structure for the simplex range
searching problem. In this problem we preprocess a set of points $P$ in $\mathbb{R}^d$
so that given a $d$-dimensional query simplex $\rho$ we can return the set of points 
$\{p \in P | p \in \rho\}$.   
Partition trees are  based on recursively partitioning the plane into a set of regions such that each
region contains at least a constant fraction of the points in $P$. 
These regions form canonical subsets and 
thus partition trees are 
%a data structure for a geometric searching problem
%where $\diamond =\, \in$ and $\gamma$ is a $d$-dimensional simplex. 
%Furthermore, they are 
%based on a decomposition scheme.  %TODO \emph{
a canonical semigroup range searching data structure %} 
Therefore, we can apply the transformation in
Theorem \ref{thm:fixed-cardinality-reportingFull},
\ref{thm:variable-cardinality-reportingFull}, \ref{thm:size-estimationFull} 
to obtain a data structure for set-difference simplex
range queries.  Note that although there has been significant research on how to
choose the partitions, Theorems \ref{thm:fixed-cardinality-reportingFull},
\ref{thm:variable-cardinality-reportingFull} and \ref{thm:size-estimationFull} can be applied independent of these
details.  We only require that the data structure  precomputes answers for each
member of a family of canonical subsets and answers queries by a combination of
these precomputed answers.
We simply replace the precomputed answe in the standard version of the data structure
for each canonical subset with the appropriate SDR or SDC sketching data structure.
Then, when answering a query, instead of summing the standard precomputed answers, we
simply add the sketches of all the canonical subsets in each query range, and then
subtract the two resulting sketches. 

%
%{\bf:} 
%Can we also apply this to the other logarithmic query time simplex range
%counting data structure based on cuttings (I think we can) and obtain a
%similar query time / space tradeoff?
%Cutting trees are also based on a partition scheme (canonical subsets) and they
%are a data structure for counting, so I'm pretty sure they work here, but I
%don't know how the combination of cutting trees and partition trees works. %
%for sake of time, forget time / space tradeoff. Maybe put in cutting trees?
%
%
%If we had more time, we could do a whole separate section on intersection
%queries. In agarwal's chapter he makes a distinction between intersection
%queries and range searching queries. Here I'm just grouping stabbing queries in
%with the range searching queries. Maybe for the journal version?
%
\subsection{Stabbing Queries}
Given a set $S$ of line segments in the plane and a vertical query line, we want to
report or count the line segments which intersect the query line. This is known
as a vertical line stabbing query.
One common solution to this problem is the segment tree.
A segment tree is an augmented binary tree, similar to the range tree.
%In the range tree, 
%Each leaf represents a small range containing only a single
%point in $S$. For a segment tree this point is one of the endpoints of a line segment $s \in
%S$. 
Each leaf represents a small range containing a single endpoint of one line
segment, 
%Also like a range tree, 
and the range $r_v$ of any interior node $v$
is %the union of the ranges of its children, or equivalently 
the union of the
ranges of its descendant leaves. Each segment $s = (x_s, y_s), (x_e, y_e)$
corresponds to the range $r(s) = [x_s, x_e]$. Let $p(v)$ denote the parent of $v$.  
We store a segment $s$ in each node $v$ such
that $r(s)$ spans $r_v$ but does not span $r_{p(v)}$.  
%Note that this implies
%that each segment is stored at most $O(\log n)$ nodes in the tree, and also that
%each segment is stored at most once on any root to leaf path in the tree. 
To answer a query we search in the segment tree for the $x$ coordinate of the 
vertical query line $r$. This root to leaf path in the tree defines a 
$O(\log n)$ size sub-family of canonical subsets of $S$, which is
exactly those canonical subsets which cover the query point $x$. 
%
%Interval trees are an alternate
%data structure which solves this problem. However while segment trees can be
%naturally generalized to higher dimensions, interval trees cannot.
%Furthermore interval trees are not based on a decomposition scheme, 
%but segment trees are based on a decomposition scheme and are therefore
%a more appropriate choice in the context of this paper. 

Hence, the segment tree is % based a decomposition scheme for the segments in $S$. %TODO \emph{
a canonical semigroup range searching data structure %} 
Thus, given a segment tree for
vertical line stabbing (counting) queries we can apply Theorems
\ref{thm:fixed-cardinality-reportingFull},
\ref{thm:variable-cardinality-reportingFull}
or \ref{thm:size-estimationFull}
to transform it into a data structure for set-difference
stabbing queries. The details are similar to the previous range searching data structures. 

\subsection{Partial Sum Queries}
% requires group, not semigroup
Let $A$ be a $d$-dimensional array for constant $d$.
A partial sum query has as input a pair of $d$-dimensional indices
$([i_1, \ldots , i_d], [j_1, \ldots j_d])$ which define a
$d$-dimensional (rectangular) range,  and to answer the query we
must return the sum of elements whose indices fall in that range.
This is similar to the orthogonal range query problem except that the points in
the set are restricted to a grid. 
%Here $\gamma$ corresponds to the set of
%indices which fall in between the query indices, $\diamond$ returns the set of
%elements with indices $\in \gamma$.  

Although the following techniques extend well to any (constant) dimension, for
clarity of exposition we describe the one and two-dimensional case.
In the one dimensional case, we must preprocess an array of length $n$ such that
for any pair of indices $1 \leq i \leq j \leq n$ we can return the sum of
elements with indices between $i$ and $j$  in constant time.
In linear time we build an array $B$ such that for any index $k$, $B[k] =
\sum_{i=1}^k A[i]$. 
Then to answer a query $Q((i,j), A)$ we return $B[j] -B[i]$.

In the two dimensional case, we preprocess by initializing $B$ with the values
$B[k, \ell] = \sum_{i=1}^k \sum_{j=1}^\ell A[k,\ell]$. This also runs in linear
time (with respect to the size of the input array $A$). 
To answer a query $Q(([i,j], [k,\ell]), A)$ we return 
$B[k, \ell] - B[i, \ell] - B[k, j] + B[i,j]$. 
See Figure~\ref{figPrefixSums} for an illustration.

In the above data structure, we can think of the array $A$ as a set of grid
cells, and each cell of $B$ represents a canonical subset of grid cell of
$A$. 
The data structure is a canonical group range searching structure which 
requires us to take  %TODO \emph{canonical semigroup range searching data structure} 
set differences of canonical subsets. This is not a problem when the elements
corresponding to the grid cells are drawn from a group, since we can just
subtract the sum of the subsets which we want to subtract.  
We can apply Theorems \ref{thm:fixed-cardinality-reportingFull},
\ref{thm:variable-cardinality-reportingFull} or \ref{thm:size-estimationFull}
%by replacing the sum at each
%$B[i,j]$ with an IBF, hierarchy of IBFs, or strata estimator over the elements in the sum. 
%Thus we 
to obtain an efficient structure for set-difference queries on a contiguous
sub-array.

\begin{figure}[h!tb]
\centering
\includegraphics[scale=.4]{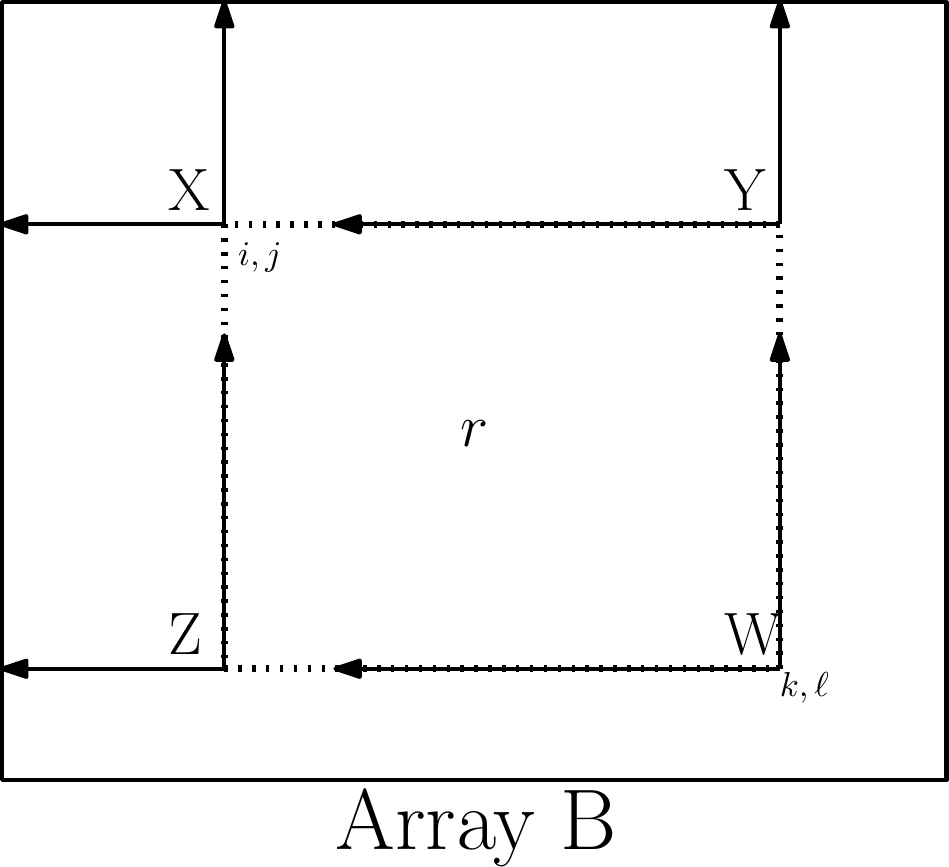}
\caption{
\label{figPrefixSums}%
A data structure for
partial sum queries when subtraction is allowed.
The prefix sums 
required to answer any partial sum query can be precomputed in constant time per
array element. Given these precomputed canonical subsets, we can answer a
partial sum query in constant time (assuming $d$ is constant). We illustrate
this concept in the case where $d = 2$.
The region $X$ is all the elements of $B$ with
indices $<i,<j$. $Y$ is all elements of $B$ with indices  $<i, < k$. $Z$ is
all elements with indices $<k, <j$. $Z$ is all elements with indices
$\leq k, \leq \ell$. The query rectangle $r$ is all elements between
indices $i,j$ and $k, \ell$ inclusive. Thus  $\sum w(r) = \sum w(W) - \sum
w(Z) - \sum w(Y) + \sum w(X)$. }
\end{figure}

\section{Conclusion and Open Problems}
We have given a general framework for converting canonical
range-searching data structures into data structures for answering
set-difference range queries, with efficient time and space bounds.

Some interesting open problems include the following:
\begin{itemize}
\item
Is it possible to
extend the results of this paper to difference range queries
on strings where distance is measured using Smith-Waterman (edit)
distance?
\item
What are some time and space lower
bounds for set-difference range queries?
\item
Can set-difference range queries be efficiently constructed for data
structures, like BBD-trees~\cite{mp-ddsar-10}
and BAR-trees~\cite{dgk-bartt-00,dgk-bartc-01}, which can be used to answer
approximate range queries?
\end{itemize}
\fi{}

\ifFull
\subsection*{Acknowledgments}
The authors' research was supported in part by NSF grant
0830403 and by the Office of Naval Research under grant
N00014-08-1-1015.
\fi{}

\ifFull{}
\clearpage
{
\raggedright
\bibliographystyle{abuser}
\bibliography{bloom}
}
\else
{
\raggedright
%\bibliographystyle{abuser}
%{\scriptsize
\bibliographystyle{abbrv}
%Needs to be merged with Range set-diff bibliographies
\bibliography{bloom}
}

\clearpage
\Fulltrue
\begin{appendix}
\section{Figures}
\label{app:figures}
\begin{figure}[h!tb]
\centering
\includegraphics[scale=.4]{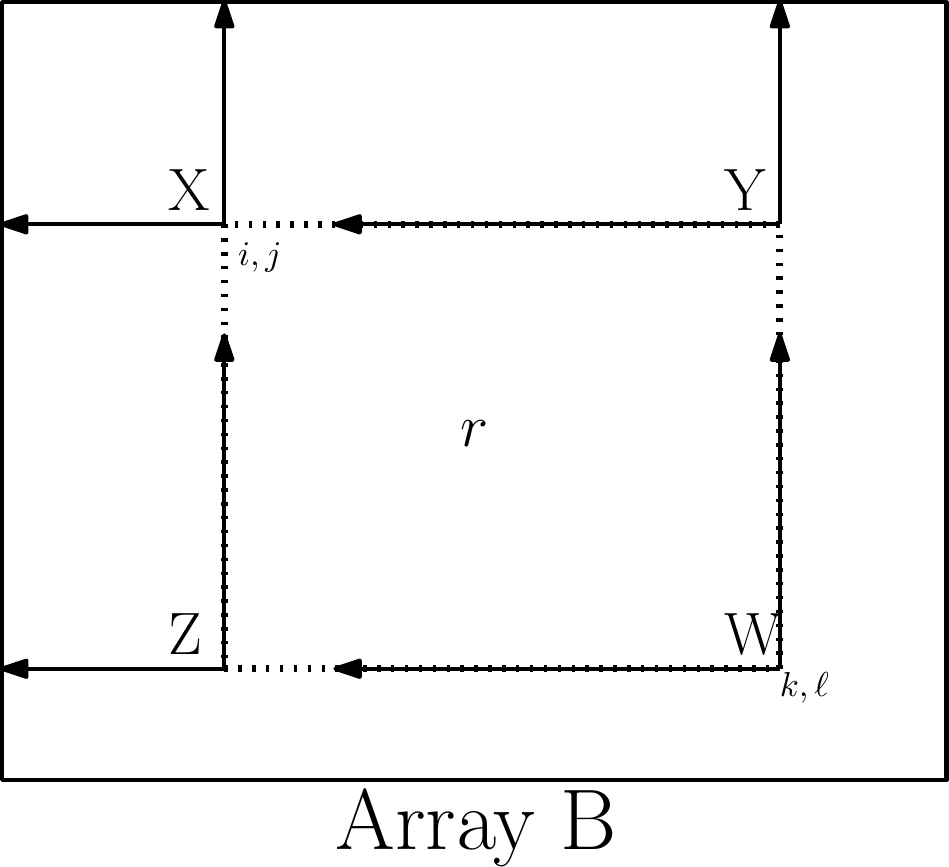}
\caption{
\label{figPrefixSums}%
A data structure for
partial sum queries when subtraction is allowed.
The prefix sums 
required to answer any partial sum query can be precomputed in constant time per
array element. Given these precomputed canonical subsets, we can answer a
partial sum query in constant time (assuming $d$ is constant). We illustrate
this concept in the case where $d = 2$.
The region $X$ is all the elements of $B$ with
indices $<i,<j$. $Y$ is all elements of $B$ with indices  $<i, < k$. $Z$ is
all elements with indices $<k, <j$. $Z$ is all elements with indices
$\leq k, \leq \ell$. The query rectangle $r$ is all elements between
indices $i,j$ and $k, \ell$ inclusive. Thus  $\sum w(r) = \sum w(W) - \sum
w(Z) - \sum w(Y) + \sum w(X)$. }
\end{figure}
\section{IBF Pseudo-Code}
\label{app:IBF-code} %label is old, but keeps links from breaking
In this section we include pseudo-code illustrating the IBF operations.

%% Including figure in the main body instead
%\begin{figure}[ht!]
%\centering
%\includegraphics[width=\columnwidth, trim = 0.1in 0.75in 0.25in 0.75in, clip]{differences.pdf}
%\caption{\label{fig:differences} Illustrating the set-difference
%range query problem. The images in (a) and (b) have four major
%differences, three of which are inside the common query range.
%The image (a) is a public-domain engraving of an astronomer by Albrecht
%D{\"u}rer, from the title page of 
%\textit{Messahalah, De scientia motus orbis} (1504).}
%\end{figure}
%
%
%
%added minipages to layout in two columns and compress space
\ifStrata  %%% should we keep this figure in?
\vspace{1.246in}
\begin{figure}[hb]% Placing it "h" makes more sense logically but takes more room
%\hspace{.1in}
%\begin{minipage}[t]{.42\textwidth}
\noindent{\sf strata}$(S)$:
  \begin{algorithmic}
    \FOR{each element $x$ in $S$}
    \STATE $y = h(x)$; $\ell = \log_2 (y \,\,\&\,\,\sim\! (y - 1))$
    \STATE add $x$ to the IBF for layer $\ell$
    \ENDFOR
  \end{algorithmic}
%\end{minipage}
%\hspace{.2in}
%\begin{minipage}[t]{.58\textwidth}
\medskip
\noindent{\sf estimate-distance}$(S,T)$:
\begin{algorithmic}
\STATE allocate storage for an IBF $C$
\STATE $d=0$
\FOR{$i=\log_2 U, \log_2 U -1,\dots, 2,1, 0$}
\STATE {\sf subtract}(layer $i$ of $S$, layer $i$ of $T$, $C$)
\IF{$C$ can be completely decoded}
\STATE add the number of decoded elements to $d$
\ELSE
\STATE return $2^{i+1}\times d$
\ENDIF
\ENDFOR 
\STATE return $d$
\end{algorithmic}
%\end{minipage}
\caption{Constructing a strata estimator, and estimating distances from two strata estimators.}
\ifFull \label{fig:strataFull} \else \label{fig:strata} \fi{} 
\end{figure}
\fi{}
\ifFull
%\begin{figure}[hbt]

\vspace{1em}
\noindent
{\sf insert}$(x)$:
\begin{algorithmic}[100]
\FOR {each $h_i(x)$ value, for $i=1,\ldots,k$}
\STATE
add $1$ to $T[h_i(x)].\mbox{\textsf{count}}$
\STATE
add $x$ to $T[h_i(x)].\mbox{\textsf{idSum}}$
\STATE
add $g(x)$ to $T[h_i(x)].\mbox{\textsf{gSum}}$
\ENDFOR
\end{algorithmic}

\medskip
\noindent
{\sf delete}$(x)$:
\begin{algorithmic}[100]
\FOR {each $h_i(x)$ value, for $i=1,\ldots,k$}
\STATE
subtract $1$ from $T[h_i(x)].\mbox{\textsf{count}}$
\STATE
subtract $x$ from $T[h_i(x)].\mbox{\textsf{idSum}}$
\STATE
subtract $g(x)$ from $T[h_i(x)].\mbox{\textsf{gSum}}$
\ENDFOR
\end{algorithmic}

\medskip
\noindent
{\sf isMember}$(x)$:
\begin{algorithmic}[100]
\FOR {each $h_i(x)$ value, for $i=1,\ldots,k$}
\IF {$T[h_i(x)].\mbox{\textsf{count}}= 0$ \textbf{and}
      $T[h_i(x)].\mbox{\textsf{idSum}}= 0$  \textbf{and}
      $T[h_i(x)].\mbox{\textsf{gSum}}= 0$}
\STATE 
\textbf{return} \textbf{false}
\ELSIF {$T[h_i(x)].\mbox{\textsf{count}}= 1$ \textbf{and}
      $T[h_i(x)].\mbox{\textsf{idSum}}= x$  \textbf{and}
      $T[h_i(x)].\mbox{\textsf{gSum}}= g(x)$ }
\STATE 
\textbf{return} \textbf{true}
\ENDIF
\ENDFOR
\STATE \textbf{return} ``not determined''
\end{algorithmic}

\medskip
\noindent
{\sf subtract}$(A,B,C)$:
\begin{algorithmic}[100]
\FOR {$i=0$ to $t-1$}
\STATE $T_C[i].\mbox{\textsf{count}} = T_A[i].\mbox{\textsf{count}} - T_B[i].\mbox{\textsf{count}}$
\STATE $T_C[i].\mbox{\textsf{idSum}} = T_A[i].\mbox{\textsf{idSum}} - T_B[i].\mbox{\textsf{idSum}}$
\STATE $T_C[i].\mbox{\textsf{gSum}} = T_A[i].\mbox{\textsf{gSum}} - T_B[i].\mbox{\textsf{gSum}}$
\ENDFOR
\end{algorithmic}

\medskip
\noindent {\sf listItems}$()$:
\begin{algorithmic}[100]
\WHILE {there is an $i\in[1,t]$ such that 
$T[i].\mbox{\textsf{count}} = 1$ \textbf{or} $T[i].\mbox{\textsf{count}} = -1$}
\IF {$T[h_i(x)].\mbox{\textsf{count}}\,=\, 1$ 
     \textbf{and} $T[h_i(x)].\mbox{\textsf{gSum}}\,=\, 
	     g(T[h_i(x)].\mbox{\textsf{idSum}})$ }
\STATE add the item,
$(T[i].\mbox{\textsf{idSum}})$,
to the ``positive'' output list
\STATE call {\sf delete}($T[i].\mbox{\textsf{idSum}}$)
\ELSIF {$T[h_i(x)].\mbox{\textsf{count}}\,=\, -1$ 
       \textbf{and} $-T[h_i(x)].\mbox{\textsf{gSum}}\,=\, 
       g(-T[h_i(x)].\mbox{\textsf{idSum}})$ }
\STATE add the item,
$(-T[i].\mbox{\textsf{idSum}})$,
to the ``negative'' output list
\STATE call {\sf insert}($-T[i].\mbox{\textsf{idSum}}$)
\ENDIF
\ENDWHILE
\end{algorithmic}

%\caption{Operations supported by an invertible Bloom filter.}
%\ifFull \label{fig:algsFull} \else \label{fig:algs} \fi{} 
%\end{figure}
\fi{}
%
%
%Some are now in main body, and some are for strata

%Reset theorem numbers
\makeatletter
\@addtoreset{theorem}{section}
\makeatother
\section{Proofs}
\label{app:proofs}

\begin{theorem}
\ifFull \label{thm:fixed-cardinality-reportingFull} \else \label{thm:fixed-cardinality-reporting} \fi{} 
Suppose that a fixed limit $m$ on the cardinality of the returned set
differences is known in advance of constructing the data structure, and our
queries must either report the difference if it has cardinality at most $m$, or otherwise report that it is too large.
In this case, we can answer set-difference range queries with probability at least $1 - \epsilon$
for any range space that can be modeled by a 
canonical group or semigroup range searching data structure.
Our solution stores $O(m)$ words of aggregate information per canonical set,
uses a combination operation $\oplus$ that takes time $O(m)$ per combination,
and allows the result of this combination to be decoded to yield the query
results in $O(m)$ time. If the data structure is updated, the aggregate information associated with each changed canonical set can itself be updated in constant time per change.
\end{theorem}
\begin{proof}
\ifFull
We associate with each canonical set an SDR sketch
%IBF 
for its elements that 
%is capable of performing the \textsf{listItems} operation on sets 
supports the \texttt{report} function for outputs of size up to $m$, and we
let $\oplus$ be defined by the operations of adding % and subtracting IBFs 
SDR sketches
as described in the previous section. The result of applying $\oplus$ to the 
%IBFs
SDR sketches
stored at each of the canonical sets  will be a single 
%IBF 
sketch for each of the queried ranges.
Then, to report the set-difference, we call the 
%\textsf{listItems} operation on this IBF, and
\texttt{report} function which
%report the items listed 
outputs the elements if there are fewer than $m$ of them. The SDR sketch is implemented as
an IBF, and the \texttt{report} function is supported via an IBF subtraction followed by
the \textsf{listItems} operation. See Section~\ref{sec:IBF} for details of the implementation. 
If the
\textsf{listItems} operation fails to decode the IBF, or if it lists more than
$m$ items, we report that the cardinality bound is exceeded.
\else
See appendix \ref{app:proofs}.
\fi
\end{proof}
\begin{theorem}
\ifFull \label{thm:variable-cardinality-reportingFull} \else \label{thm:variable-cardinality-reporting} \fi{} 
Suppose that we wish to report set differences that may be large or small,
without the fixed bound $m$, in a time bound that depends on the size of the
difference but that may depend only weakly on the size of the total data set. In
this case, we can answer range difference queries with probability at least $1 - \epsilon$ 
for any range space that can
be modeled by a canonical group or semigroup range searching data structure. Our
solution stores a number of words of aggregate information per canonical set
that is $O(1)$ per element of the set, uses a combination operation $\oplus$
that takes time $O(m)$ (where $m$ is the cardinality of the final set-theoretic
difference) and allows the result of this combination to be decoded to yield the
query results in $O(m)$ time. If the data structure is updated, the aggregate
information associated with each changed canonical set can itself be updated in
logarithmic time per change.  
\end{theorem}
\begin{proof}
\ifFull
For each canonical set $S$ of the range query data structure, we store an SDR sketch
consisting of a hierarchy of $\log_2
|S|$ IBFs, capable of successfully performing the \textsf{listItems} operation
on sets of sizes that are powers of two up to the size of $S$ itself. The total
size of these IBFs adds in a geometric series to be proportional to the size of
$S$. The time to update the data structure then follows since each element of
each canonical set is stored in a logarithmic number of IBFs.

To answer a range query, suppose first that $m$ is known. In this case, we
choose $j=\lceil\log_2 d\rceil$, so that $2^j>m$ but $2^j=O(m)$, and we collect
and combine the IBFs of size $2^j$ exactly as in
Theorem~\ifFull{\ref{thm:fixed-cardinality-reportingFull}}\else{\ref{thm:fixed-cardinality-reporting}}\fi{}. The only possible complication is
that, for some canonical sets, the number $2^j$ may be larger than the number of
elements in the canonical set, so that we have no pre-stored IBF of the correct
size to collect. In this case, though, we can construct an IBF of the correct
size from the list of elements of the canonical set, in time $O(2^j)$.
Finally, if $m$ is unknown, we try the same procedure for successive powers of two until finding a set of powers for which the result succeeds.
\else
See Appendix \ref{app:proofs}.
\fi
\end{proof}
\begin{theorem}
\ifFull \label{thm:size-estimationFull} \else \label{thm:size-estimation} \fi{} 
Suppose that we wish to report the cardinality of the set difference rather than
its elements, and further that we allow this cardinality to be reported
approximately, within a fixed approximation ratio that may be arbitrarily close
to one.  In this case, we can answer range difference queries with 
probability at least $1 - \epsilon$ for any range
space that can be modeled by a canonical group or semigroup range searching data
structure.  Our solution stores a number of words of aggregate information per
canonical set that has size $O(\log n \log U)$, uses a combination operation $\oplus$
that takes time $O(\log n \log U)$ and allows the result of this combination to be
decoded to yield the query results in $O(\log n \log U)$ time. 
%\textbf{are these log n bounds right?}
\end{theorem}
\ifStrata
\begin{proof}
We associate a strata estimator with each canonical set, and we let $\oplus$ be
defined by the operations of adding and subtracting strata estimators as
described in the previous section. The result of applying $\oplus$ to these
strata estimators will be a single strata estimator for the queried range
difference, from which we may use the estimator to approximate the cardinality
of the difference. The space and time bounds are those for the strata estimator.
\end{proof}
\else
\begin{proof}
\ifFull
We associate with each canonical set an SDC sketch for its elements that supports the 
\texttt{count} function with probability $1 - \epsilon$.
and we let $\oplus$ be
defined by the operations of adding the SDC sketches.
The result of applying $\oplus$ to these
sketches will be a single sketch for each of the queried ranges.
Then to approximate the cardinality of the difference, we call the \texttt{count} function.
We implement each SDC sketch as set of $O(\log(1/\epsilon))$ independent linear 
frequency moment estimation sketches. 
The space and time bounds are the same as those for computing these sketches, and further
details of the implementation are given in Section~\ref{sec:fme}.
\else
See Appendix \ref{app:proofs}.
\fi
\end{proof}

\begin{theorem}
\ifFull \label{thm:bloom-decodeFull} \else \label{thm:bloom-decode} \fi{} 
Suppose $X$ and $Y$ are sets with $m$ elements in their
symmetric difference, i.e., $m=|X\bigtriangleup Y|$, and let
$\epsilon>0$ be an arbitrary real number.
Let $A$ and $B$ be invertible Bloom filters
built from $X$ and $Y$, respectively, such that each IBF has
$\lambda\ge k+\lceil\log k\rceil$ bits in its \textsf{gSum} field, i.e.,
the range of $g$ is $[1,2^\lambda)$,
and each IBF has at least
$2km$ cells, where $k > \lceil \log (m/\epsilon)\rceil+1$ is the
number of hash functions used.
Then the \textsf{listItems} method for the IBF $C$ resulting from the
\textsf{subtract}$(A,B,C)$ method will list all $m$
elements of $X\bigtriangleup Y$ and identify which belong
to $X\setminus Y$ and which belong to $Y\setminus X$ with 
probability at least $1-\epsilon$.
\end{theorem}
\begin{proof}
\ifFull
There are at least $km$ empty cells in $C$; hence,
the probability that an element 
$z\in X\bigtriangleup Y$ collides with another element
in $X\bigtriangleup Y$ in each of its $k$ cells
is at most $2^{-k}$, independent of 
the locations chosen in $C$ by the other 
elements in $X\bigtriangleup Y$.
Thus, the probability that any element in $X\bigtriangleup Y$ has $k$
collisions is at most $m\, 2^{-k} \le \epsilon/2$.
Therefore, with probability at least $1-\epsilon/2$, every 
element in $X\bigtriangleup Y$ has at least one cell in $C$ where it is
the only occupant.

Next, consider the probability that a cell has
$|\textsf{count}|=1$ and $|\textsf{gSum}| = g(|\textsf{idSum}|)$ but
nevertheless has two or more occupants.
Since the \textsf{gSum} field has
at least $k$ bits,
the probability that a cell is ``bad'' in this way
is at most
$2^{-\lambda}$. 
Thus, since there are at most $km$ non-zero cells,
the probability we have any such bad cell in $C$ is at
most
\[
km\, 2^{-\lambda} \le m 2^{-k} \le \epsilon/2.
\]
This completes the proof.
\else
See Appendix \ref{app:proofs}.
\fi
\end{proof}

\ifStrata

\begin{lemma}
\ifFull \label{lem:strata-sizesFull} \else \label{lem:strata-sizes} \fi{} 
For any $\epsilon>0$, with probability $1-\epsilon$, the
cardinality of each of the strata from $0$ to $i$ in $S$
is within a multiplicative factor of
$(1\pm\delta)$ of its expectation, for
$\delta \ge \sqrt{4(2+\log (1/\epsilon))(2^i/m)\ln 2}$.
\end{lemma}
\begin{proof}
\ifFull
Let $Z_j$ be a random variable representing the size of the $j$th stratum,
for $j=0,1,\ldots,i$,
and let 
$\mu_j = E(Z_j) = m/2^j$.
By standard Chernoff bounds, 
\ifFull
(e.g., see~\cite{mu-pcrap-05,mr-ra-95}),
\fi{}
for $\delta > 0$,
%using align instead compresses space,
% but why not also use it in full version
%\ifFull
%\[
%  \Pr\left(Z_j > (1+\delta)\mu_j\right) \,< \,e^{-\mu_j\delta^2/4} 
%\]
%and
%\[
%  \Pr\left(Z_j < (1-\delta)\mu_j\right)  \,<\, e^{-\mu_j\delta^2/4}.
%\]
%\else
\begin{align} 
  \Pr\left(Z_j > (1+\delta)\mu_j\right) & < e^{-\mu_j\delta^2/4} \\
  \mbox{and \quad} 
  Pr\left(Z_j < (1-\delta)\mu_j\right) & < e^{-\mu_j\delta^2/4}.
\end{align}
%\fi{}
By taking $\delta \ge \sqrt{4(2+\log (1/\epsilon))(2^i/m)\ln 2}$, 
we have that the probability that a particular $Z_j$ is
not within a multiplicative factor of
$1\pm \delta$ of its expectation is at most
\[
2e^{-\mu_j\delta^2/4} \,\le\, 2^{-(1+\log (1/\epsilon))2^{i-j}} 
   \, =\, \frac{(\epsilon/{2})}{2^{2^{i-j}}}.
\]
Thus, by a union bound, the probability that any $Z_j$
is not within a multiplicative factor of $1\pm \delta$ of its
expectation is at most
\[
\sum_{j=0}^i \frac{(\epsilon/{2})}{2^{2^{i-j}}}
\,
= 
\,
({\epsilon}/{2})  \sum_{j=0}^i 2^{-2^{i-j}} ,
\]
which is at most
$\epsilon$.
\else
{
    See Lemma \ref{lem:strata-sizesFull} in Appendix \ref{app:proofs}. 
}
\fi{}
\end{proof}
\begin{theorem}
\ifFull \label{thm:strataFull} \else \label{thm:strata} \fi{} 
Let $X$ and $Y$ whose symmetric difference has size $m$,
and let $0<\epsilon<1$ and $0<\delta<1$ be arbitrary real numbers.
Suppose we encode $X$ and $Y$ in a strata estimator, $S$, of height
$\lceil\log U\rceil$, 
with each member IBF having $2km'$ cells, $k$ hash functions, and 
\textsf{gSum} size of $k+\lceil\log k\rceil$ bits, for 
$k\ge \lceil\log 2m'/\epsilon\rceil+1$ and 
$m' \ge 4\delta^{-2}(2+\log (2/\epsilon))\ln 2$,
which is $\Omega(\delta^{-2}\log (1/\epsilon))$. 
Then, with probability at least
$1-\epsilon$, the estimate resulting from our decoding algorithm,
run on the first IBF that is not too full to decode, will
be within a factor of $1\pm\delta$ of $m$.
\end{theorem}

\begin{proof}
\ifFull
{
By Lemma~\ifFull{\ref{lem:strata-sizesFull}}\else{\ref{lem:strata-sizes}}\fi{},
for $\delta = \sqrt{4(2+\log (2/\epsilon))(2^i/m)\ln 2}$, 
the cardinalities of the first $i$ levels of a strata
estimator are each within a factor of $1\pm\delta$ 
of its expectation, $m/2^i$, with probability $1-\epsilon/2$.
We don't know the value of $m$ when we construct the
strata estimators, but, by Theorem~\ifFull{\ref{thm:bloom-decodeFull}}\else{\ref{thm:bloom-decode}}\fi{}, 
if each IBF has $2km'$ cells, 
where $k = \lceil \log 2m'/\epsilon\rceil+1$
and $m' \ge 2$, then we can decode a set of
$m'$ elements with probability at least $1-\epsilon/2$.
Thus, we can assume without loss of generality that $m>m'$, for otherwise with high probability we will 
decode the entire strata estimator and determine the correct distance without error.
With probability $1-\epsilon$,
we can achieve an approximation factor of $1\pm\delta$ for our size
estimator by taking
\ifFull
\[
m' \ge 4\delta^{-2}(2+\log (2/\epsilon))\ln 2. 
\]
\else
$m' \ge 4\delta^{-2}(2+\log (2/\epsilon))\ln 2. $
\fi{}
Depending on the level, $i$, where we can first successfully decode
the contents of its IBF,
this gives us a $(1\pm\delta)$ estimate for $m/2^i$, which we convert
into a $(1\pm\delta)$ estimate for $m$ by multiplying by $2^i$.
}
\else
{
    See Theorem \ref{thm:strata} Appendix \ref{app:proofs}. 
}
\fi{}
\end{proof}
\fi{}

%We did the general framework, now lets show it off for specific examples
%\section{Set-Difference Range Searching}
\section{Set-Difference Range Searching Details}
\label{sec:sdDetails}
\ifFull
%Many data structures for designed to answer geometric searching queries 
%are based on a decomposition scheme.
In this section, we look at several examples of canonical semigroup range searching
structures and show how to transform them into data structures that can
efficiently answer set-difference range queries.  
Specifically, we give more detail for the results which were summarized in 
Table~\ref{tab:summary}. 

%
%NOTE: do to our notation earlier, there isn't really a distinction between
%intersection and range searching
%
%We examine two special cases of geometric searching problems. 
%First we examine data structures for \emph{range searching}
%in which the result of a query 
%Then we examine a data structure for \emph{intersection searching} in which 
%$\diamond(\gamma,S)$ is the subset of $S$ which intersects $\gamma$.
%Regardless of the specific definition of $\diamond$, the data structures
%are based on decomposition schemes. 
%Therefore we can apply Theorem \ref{thmTransform} to  
%transform them into a data structures for 
%set-difference range searching or set-difference intersection searching.  
%
\subsection{Orthogonal Range Searching}
In the orthogonal range searching problem, our task is to preprocess a set of points in $R^d$
 such that given a query range defined by an axis-parallel hyper rectangle, we can efficiently report or count the set of points contained in the hyper-rectangle
The range tree is a classic canonical semigroup range searching data structure that
solves this problem.

In one dimension, a range tree is a binary tree in which each leaf represents a
small range containing a single point, and an in-order traversal of the leaves
corresponds to the sorted order of those points. The range of any interior node
is the union of the
ranges of its descendant leaves. 
Each node stores the value of the aggregate function on its 
canonical set---the
set of points in its range.

%Thus the range of the root node contains all the points stored in the range tree. 
%In a range tree for the orthogonal range
%counting problem, each node stores the sum of the weights of the points in its
%range. 

We can easily build a $d \geq 2$ dimensional range tree by building a range tree
over the $d^{\textnormal{th}}$ dimension, and at each internal node we store a
($d-1$)-dimensional dimensional range tree over the points in that node's range.
To answer a query using a range tree, we find a set of $O(\log^d n)$ canonical
ranges which exactly cover the query range and report the sum of the 
aggregate functions of each of these canonical ranges. 

Thus  
we can apply Theorem \ref{thm:fixed-cardinality-reporting},
\ref{thm:variable-cardinality-reporting}, or \ref{thm:size-estimation} 
 to obtain a data structure for set-difference orthogonal range queries.
We preprocess each canonical set into an appropriate sketch, and instead of storing the
value of the aggregate function in each node, instead we store a pointer to the
corresponding sketch for that node's canonical set. When answering a query, instead of
summing the aggregate functions, we simply add or subtract the sketches.

%Note that although it is common to apply fractional cascading techniques to gain
%a $O(\log n)$ speedup, we cannot apply the above transformation in conjunction
%with fractional cascading techniques. However, we do gain substantial savings in
%both space and query time when the size of the set difference is expected to be
%small. 
%
\subsection{Simplex Range Searching}
Partition trees~\cite{Matousek1992} 
are a linear space data structure for the simplex range
searching problem. In this problem we preprocess a set of points $P$ in $\mathbb{R}^d$
so that given a $d$-dimensional query simplex $\rho$ we can return the set of points 
$\{p \in P | p \in \rho\}$.   
Partition trees are  based on recursively partitioning the plane into a set of regions such that each
region contains at least a constant fraction of the points in $P$. 
These regions form canonical subsets and 
thus partition trees are 
%a data structure for a geometric searching problem
%where $\diamond =\, \in$ and $\gamma$ is a $d$-dimensional simplex. 
%Furthermore, they are 
%based on a decomposition scheme.  %TODO \emph{
a canonical semigroup range searching data structure %} 
Therefore, we can apply the transformation in
Theorem \ref{thm:fixed-cardinality-reporting},
\ref{thm:variable-cardinality-reporting}, \ref{thm:size-estimation} 
to obtain a data structure for set-difference simplex
range queries.  Note that although there has been significant research on how to
choose the partitions, Theorems \ref{thm:fixed-cardinality-reporting},
\ref{thm:variable-cardinality-reporting} and \ref{thm:size-estimation} can be applied independent of these
details.  We only require that the data structure  precomputes answers for each
member of a family of canonical subsets and answers queries by a combination of
these precomputed answers.
We simply replace the precomputed answe in the standard version of the data structure
for each canonical subset with the appropriate SDR or SDC sketching data structure.
Then, when answering a query, instead of summing the standard precomputed answers, we
simply add the sketches of all the canonical subsets in each query range, and then
subtract the two resulting sketches. 

%
%{\bf:} 
%Can we also apply this to the other logarithmic query time simplex range
%counting data structure based on cuttings (I think we can) and obtain a
%similar query time / space tradeoff?
%Cutting trees are also based on a partition scheme (canonical subsets) and they
%are a data structure for counting, so I'm pretty sure they work here, but I
%don't know how the combination of cutting trees and partition trees works. %
%for sake of time, forget time / space tradeoff. Maybe put in cutting trees?
%
%
%If we had more time, we could do a whole separate section on intersection
%queries. In agarwal's chapter he makes a distinction between intersection
%queries and range searching queries. Here I'm just grouping stabbing queries in
%with the range searching queries. Maybe for the journal version?
%
\subsection{Stabbing Queries}
Given a set $S$ of line segments in the plane and a vertical query line, we want to
report or count the line segments which intersect the query line. This is known
as a vertical line stabbing query.
One common solution to this problem is the segment tree.
A segment tree is an augmented binary tree, similar to the range tree.
%In the range tree, 
%Each leaf represents a small range containing only a single
%point in $S$. For a segment tree this point is one of the endpoints of a line segment $s \in
%S$. 
Each leaf represents a small range containing a single endpoint of one line
segment, 
%Also like a range tree, 
and the range $r_v$ of any interior node $v$
is %the union of the ranges of its children, or equivalently 
the union of the
ranges of its descendant leaves. Each segment $s = (x_s, y_s), (x_e, y_e)$
corresponds to the range $r(s) = [x_s, x_e]$. Let $p(v)$ denote the parent of $v$.  
We store a segment $s$ in each node $v$ such
that $r(s)$ spans $r_v$ but does not span $r_{p(v)}$.  
%Note that this implies
%that each segment is stored at most $O(\log n)$ nodes in the tree, and also that
%each segment is stored at most once on any root to leaf path in the tree. 
To answer a query we search in the segment tree for the $x$ coordinate of the 
vertical query line $r$. This root to leaf path in the tree defines a 
$O(\log n)$ size sub-family of canonical subsets of $S$, which is
exactly those canonical subsets which cover the query point $x$. 
%
%Interval trees are an alternate
%data structure which solves this problem. However while segment trees can be
%naturally generalized to higher dimensions, interval trees cannot.
%Furthermore interval trees are not based on a decomposition scheme, 
%but segment trees are based on a decomposition scheme and are therefore
%a more appropriate choice in the context of this paper. 

Hence, the segment tree is % based a decomposition scheme for the segments in $S$. %TODO \emph{
a canonical semigroup range searching data structure %} 
Thus, given a segment tree for
vertical line stabbing (counting) queries we can apply Theorems
\ref{thm:fixed-cardinality-reporting}, \ref{thm:variable-cardinality-reporting}
or \ref{thm:size-estimation}
to transform it into a data structure for set-difference
stabbing queries. The details are similar to the previous range searching data structures. 

\subsection{Partial Sum Queries}
% requires group, not semigroup
Let $A$ be a $d$-dimensional array for constant $d$.
A partial sum query has as input a pair of $d$-dimensional indices
$([i_1, \ldots , i_d], [j_1, \ldots j_d])$ which define a
$d$-dimensional (rectangular) range,  and to answer the query we
must return the sum of elements whose indices fall in that range.
This is similar to the orthogonal range query problem except that the points in
the set are restricted to a grid. 
%Here $\gamma$ corresponds to the set of
%indices which fall in between the query indices, $\diamond$ returns the set of
%elements with indices $\in \gamma$.  

Although the following techniques extend well to any (constant) dimension, for
clarity of exposition we describe the one and two-dimensional case.
In the one dimensional case, we must preprocess an array of length $n$ such that
for any pair of indices $1 \leq i \leq j \leq n$ we can return the sum of
elements with indices between $i$ and $j$  in constant time.
In linear time we build an array $B$ such that for any index $k$, $B[k] =
\sum_{i=1}^k A[i]$. 
Then to answer a query $Q((i,j), A)$ we return $B[j] -B[i]$.

In the two dimensional case, we preprocess by initializing $B$ with the values
$B[k, \ell] = \sum_{i=1}^k \sum_{j=1}^\ell A[k,\ell]$. This also runs in linear
time (with respect to the size of the input array $A$). 
To answer a query $Q(([i,j], [k,\ell]), A)$ we return 
$B[k, \ell] - B[i, \ell] - B[k, j] + B[i,j]$. 
See Figure~\ref{figPrefixSums} for an illustration.

In the above data structure, we can think of the array $A$ as a set of grid
cells, and each cell of $B$ represents a canonical subset of grid cell of
$A$. 
The data structure is a canonical group range searching structure which 
requires us to take  %TODO \emph{canonical semigroup range searching data structure} 
set differences of canonical subsets. This is not a problem when the elements
corresponding to the grid cells are drawn from a group, since we can just
subtract the sum of the subsets which we want to subtract.  
We can apply Theorems \ref{thm:fixed-cardinality-reporting},
\ref{thm:variable-cardinality-reporting} or \ref{thm:size-estimation}
%by replacing the sum at each
%$B[i,j]$ with an IBF, hierarchy of IBFs, or strata estimator over the elements in the sum. 
%Thus we 
to obtain an efficient structure for set-difference queries on a contiguous
sub-array.

\end{appendix}
\fi{}
\fi{}
\fi{}

\end{document}